\documentclass[acmsmall
             ,nonacm
             ]{acmart}

\usepackage{amsmath,amsthm,amsbsy}
\usepackage{subcaption}
\usepackage{mathtools}
\usepackage{microtype}

\usepackage{setspace}
\usepackage{color}
\usepackage{xcolor}
\usepackage{multirow}
\usepackage{paralist}
\usepackage{xspace}
\usepackage[ruled, lined, linesnumbered, commentsnumbered, longend]{algorithm2e}
\usepackage{tikz}
\usetikzlibrary{
  graphs,
  graphs.standard
}
\usepackage{tabularx,booktabs}

\let\originalleft\left
\let\originalright\right
\renewcommand{\left}{\mathopen{}\mathclose\bgroup\originalleft}
\renewcommand{\right}{\aftergroup\egroup\originalright}

\newcommand{\ra}{\rightarrow}

\renewcommand{\Pr}{\operatorname*{\textbf{\textup{Pr}}}}
\DeclareMathOperator*{\II}{\textbf{\textup{I}}}

\DeclareMathOperator*{\HH}{\textbf{\textup{H}}}

\DeclareMathOperator*{\EE}{\textbf{\textup{E}}}

\renewcommand{\geq}{\geqslant}
\renewcommand{\ge}{\geqslant}

\renewcommand{\le}{\leqslant}
\newcommand{\set}[1]{\{ #1 \}}

\newcommand{\lt}{\left}
\newcommand{\rt}{\right}

\newcommand{\congest}{\ensuremath{\mathsf{CONGEST}}}

\makeatletter
\newtheorem*{rep@theorem}{\rep@title}
\newcommand{\newreptheorem}[2]{%
\newenvironment{rep#1}[1]{%
\def\rep@title{#2 \ref{##1}}%
\begin{rep@theorem}[restated]}%
{\end{rep@theorem}}}
\makeatother

\newreptheorem{lemma}{Lemma}
\newreptheorem{theorem}{Theorem}

\newboolean{short}
\setboolean{short}{false}
\newcommand{\onlyShort}[1]{\ifthenelse{\boolean{short}}{#1}{}}
\newcommand{\onlyLong}[1]{\ifthenelse{\boolean{short}}{}{#1}}

\newtheorem{invariant}{Invariant}
\theoremstyle{plain}
\theoremstyle{plain}
\newtheorem{fact}{Fact}

\newcommand{\md}{\middle}
\newcommand{\ann}[1]{%
\text{\footnotesize(#1)}\quad}

\usepackage[most]{tcolorbox}
\tcbuselibrary{theorems}
\tcolorboxenvironment{theorem}{lower separated=false,
                colback=white!90!gray,
colframe=white, fonttitle=\bfseries,
colbacktitle=white!50!gray,
coltitle=black,
enhanced,
boxed title style={colframe=black},
}
\tcolorboxenvironment{corollary}{lower separated=false,
                colback=white!90!gray,
colframe=white, fonttitle=\bfseries,
colbacktitle=white!50!gray,
coltitle=black,
enhanced,
boxed title style={colframe=black},
}

\usepackage{thmtools}
\usepackage{thm-restate}
\usepackage[justification=centering]{caption} %

\newcommand{\free}{f} 
\newcommand{\msg}[1]{\langle \textsf{#1} \rangle} 
 
\newcommand{\spreading}{\textsf{Spreading}} 
\newcommand{\Good}{\mathsf{Lrg}}

\title{Dynamic Maximal Matching in Clique Networks} %

\author{Minming Li}
\affiliation{
  \department{Department of Computer Science}
  \institution{City University of Hong Kong}
  \country{Hong Kong SAR}
}

\author{Peter Robinson}
\affiliation{
  \department{School of Computer \& Cyber Sciences}
  \institution{Augusta University}
  \country{Georgia, USA}
}

\author{Xianbin Zhu}
\affiliation{
  \department{Department of Computer Science}
  \institution{City University of Hong Kong}
  \country{Hong Kong SAR}
}


\begin{document}

\begin{abstract}
We consider the problem of computing a maximal matching with a distributed algorithm in the presence of batch-dynamic changes to the graph topology. 
We assume that a graph of $n$ nodes is vertex-partitioned among $k$ players that communicate via message passing. 
Our goal is to provide an efficient algorithm that quickly updates the matching even if an adversary determines batches of $\ell$ edge insertions or deletions. 

Assuming a link bandwidth of $O(\beta\log n)$ bits per round, for a parameter $\beta \ge 1$, we first show a lower bound of $\Omega\lt( \frac{\ell\,\log k}{\beta\,k^2\log n}\rt)$ rounds for recomputing a matching assuming an oblivious adversary who is unaware of the initial (random) vertex partition as well as the current state of the players, and a stronger lower bound of $\Omega(\frac{\ell}{\beta\,k\log n})$ rounds against an adaptive adversary, who may choose any balanced (but not necessarily random) vertex partition initially and who knows the current state of the players. 

We also present a randomized algorithm that has an initialization time of $O \lt( \lceil\frac{n}{\beta\,k}\rceil\log n \rt)$ rounds, while achieving an update time that that is independent of $n$:
In more detail, the update time is $O\lt( \lceil \frac{\ell}{\beta\,k} \rceil \log(\beta\,k) \rt)$ against an oblivious adversary, who must fix all updates in advance.
If we consider the stronger adaptive adversary, the update time becomes $O\lt( \lceil \frac{\ell}{\sqrt{\beta\,k}}\rceil \log(\beta\,k) \rt)$ rounds.
\end{abstract}
\maketitle

\addtocontents{toc}{\protect\setcounter{tocdepth}{2}}

\section{Introduction} \label{sec:intro}
Designing efficient algorithms for large graphs is crucial for many applications, ranging from transportation networks to protein-interaction networks in biology.
Real-world graph data sets are inherently dynamic in the sense that the nodes and edges between them may change over time. 
While much of the previous literature on dynamic updates in the centralized setting consider only a single update at a time, in a distributed network, each node runs an instance of a distributed algorithm and thus, if multiple updates are happening (almost) simultaneously perhaps affecting different parts of the network, it is desirable to process all of these updates in the same batch to reduce the overall time complexity. 
In this work, we study the fundamental problem of computing a maximal matching with a distributed algorithm under batch-dynamic edge updates. 
A \emph{matching} $M$ in a graph $G$ is a set of edges without common vertices, and we say that $M$ is a \emph{maximal matching} if no matching in $G$ is a proper superset of $M$.
We point out that, recently, batch-dynamic parallel algorithms have received significant attention, e.g., see \cite{acar2019parallel, acar2011parallelism, anderson2023parallel, dhulipala2020parallel}.

\medskip\noindent\textbf{The $k$-Clique Message Passing Model.}
We assume a communication network that is a clique of $k$ players $P_1,\dots,P_k$, each of whom is executing an instance of a distributed algorithm. 
The input is an $n$-vertex graph, where each vertex is labeled with a unique integer ID from $[n]$. 
For technical reasons that will become clear in Section~\ref{sec:dyn-mm-ran}, we assume that $k=\Omega\lt( \log^4n \rt)$, and we are mostly interested in the case where $k\ll n$, which captures the realistic setting that the amount of data exceeds the number of available processing units by far.

The players communicate via message passing over the bidirectional links of the clique network, and we assume that the computation proceeds in \emph{synchronous rounds}: 
In every round, a player performs some local computation, which may include private random coin flips, and then sends a message of at most $O(\beta\log n)$ bits over each one of its $k-1$ incident communication links, where the integer $\beta \ge 1$ is the \emph{bandwidth parameter} of the model; note that $\beta$ is not necessarily a constant.
All messages are guaranteed to be received by the end of the same round, and, in each round, every player can receive up to $O(\beta\,k\log n)$ bits in total from the other players. 
As we consider distributed algorithms, each player initially only knows part of the input.  
Here we consider \textbf{vertex-partitioning}, which means that each vertex $u$ is assigned to some player $P$.
For each vertex $v$ assigned to $P$, we say that $P$ \emph{hosts} vertex $v$. We use the notation $P(v)$ to denote the player hosting $v$, and use $V(P)$ to denote the set of vertices hosted by $P$. Note that vertices cannot migrate between players.

\medskip
\noindent\textbf{Adversary and Input Assignment.} 
The input of a player $P$ consists of the IDs of its hosted vertices $V(P)$ and, for each $v \in V(P)$, player $P$ knows the IDs of $v$'s neighbors and the players to which they were assigned. 
In this work, we consider two adversarial models:
The \emph{adaptive adversary} may fix any \emph{balanced vertex partitioning}, which requires each player to hold at least $\Omega\lt( \frac{n}{k} \rt)$ and at most $O\lt( \frac{n}{k}\log n \rt)$ vertices in total. 
On the other hand, when considering the \emph{oblivious adversary}, we assume \emph{random vertex partitioning}, which means that each vertex is uniformly at random assigned to any player.
We point out that randomly partitioning the vertices is a common assumption in real-world graph processing systems, e.g.,  see~\cite{facebookgiraph}.
In expectation, each player $P$ obtains a set of $\Theta(n/k)$ vertices, and thus random vertex partitioning yields a balanced partitioning with high probability.

\medskip\noindent\textbf{Dynamic Changes, Initialization, and Output.}
We model dynamical changes in this setting by considering a sequence of updates to the edges of the graph. 
Each update consists of a \emph{batch} of at most $\ell$ edge deletions and  insertions, chosen by the adversary as follows:
\begin{compactitem}
\item The adaptive adversary may observe the local memory of the players, which includes the currently computed matching, before choosing the next batch of $\ell$ edge updates. 
\item The oblivious adversary must choose the entire sequence of batch updates in advance. %
\end{compactitem}
Since we are assuming a distributed algorithm, each player only needs to output the part of the solution relevant to its hosted vertices. 
In more detail, each player $P$ needs to output every edge $\set{u,v} \in M$, for which it hosts an endpoint, e.g., $u \in V(P)$.

We allow the algorithm to perform some initial computation right after the vertex partitioning and before the very first batch of updates arrives. %
The number of rounds required for this part defines the \emph{initialization time} of the algorithm.
Then, the first batch of updates takes place instantaneously, and every player that hosts a vertex $v$ incident to an added or deleted edge learns about $v$'s new neighborhood. 
These changes yield a modified graph $G'$. 
The algorithm must react by updating $M$ to yield a maximal matching for $G'$; no further changes to the graph topologies occur until the algorithm has completed its update.
Subsequently, the next batch of updates arrives, and so forth.

\medskip
\noindent\textbf{Update Time.}
As we are considering randomized algorithms, we give probabilistic guarantees on the complexity measures. 
When saying that an event $\mathcal{E}$ holds \emph{with high probability (in $N$)}, for some variable $N$, we mean that $\mathcal{E}$ occurs with probability at least $1 - 1/N^{c}$, where $c$ is a positive constant.
If $N=n$, we simply omit $N$ and say that $\mathcal{E}$ happens with high probability.
In particular, an \emph{algorithm has an update time of $\ T$ with high probability in $N$}, if the algorithm outputs a maximal matching following a batch of $\ell$ updates in at most $T$ rounds w.h.p.\ in $N$, assuming that the local states of the players correspond to a maximal matching on the graph prior to these updates. 
Note that, for the oblivious adversary, the probability is taken over the random coin flips of the players, as well as the random vertex partitioning.
Since a random partition is very likely going to be balanced, it is sufficient to show a bound on the update time assuming a (roughly) balanced partition.

\medskip\noindent\textbf{Relationship to $k$-Machine Model and Congested Clique.}
When assuming the oblivious adversary and $\beta = 1$, our model corresponds to the \emph{$k$-machine model}~\cite{klauck2014distributed,hourani2013distributed}, which is motivated by vertex-centric graph processing frameworks such as Google Pregel~\cite{pregel}, Apache Giraph~\cite{facebookgiraph}, and GraphX~\cite{graphx}.
Thus, our results for the oblivious case directly extend to this setting.

For $k=n$ players and $\beta=1$, our model is equivalent to the congested clique~\cite{lotker} if each player obtains exactly one vertex. 
Thus our approach also yields a dynamic algorithm that are communication-efficient when updating a matching for the latter model.
We elaborate this point in more detail in Section~\ref{sec:dmm-cc}.

\medskip\noindent\textbf{Local Memory.} Analogously to the congested clique and the $k$-machine model, we do not impose a restriction on the local memory of the players.
This stands in contrast to the popular Massively Parallel Computation (MPC) model of \cite{karloff2010model}, where the memory of each machine (i.e., player) is limited, which in turn limits the amount of information any machine can receive in a single round.

We point out that our algorithm is nevertheless space-efficient, in the sense that each player $P_i$ uses at most $O\lt(\max\set{n,\beta k,|G[V(P_i)]|}\cdot\log n\rt)$ bits of local memory, where $|G[V(P_i)]|$ is the size of the subgraph induced by $P_i$'s hosted vertices and their incident edges.

\subsection{Our Contributions}

We present the first bounds for maximal matching in the $k$-clique message passing model under batch-dynamic updates, where each batch consists of at most $\ell$ edge additions or deletions.
We start by determining lower bounds on the necessary update time of any algorithm:

\newcommand{\firstlowerbound}{
Consider any randomized algorithm for maximal matching in the $k$-clique message passing model that initially constructs a maximal matching (w.h.p.), and is guaranteed to recompute a maximal matching with an update time of $T$ (w.h.p.), assuming $\ell$ edge-updates per batch, for any $\ell\le n/2k$ and $k\le \sqrt{n}/\log n$. Then, the following hold:
	\begin{compactenum}
 \item $T = \Omega\lt( \frac{\ell\log k}{\beta\,k^2\log n} \rt)$ rounds, against an oblivious adversary;
 \item $T = \Omega\lt( \frac{\ell}{\beta\,k\log n} \rt)$ rounds, against an adaptive adversary.
  \end{compactenum}
These results hold for any number of initialization rounds, and even if the players have access to shared randomness.   
}
\begin{theorem} \label{thm:mm_lb}
	\firstlowerbound
\end{theorem}

We point out that there is a na\"ive way to obtain an update time of $O(\lceil \ell/\beta k\rceil)$ rounds that applies to any problem in this setting, if we allow a prohibitively large initialization time of $O(m/\beta k) = O(n^2/\beta k)$ rounds, for an input graph with $m$ edges.
To see why this is the case, observe that every player can learn the entire input graph within $O(m/\beta k)$ rounds, considering that an edge fits into a message of $O(\log n)$ bits and a player can receive $O(\beta k)$ such messages per round.
Then, upon any batch of $\ell$ updates, each player simply sends all its updates to every other player, which, by using a simple information dissemination technique (see Lemma~\ref{lem:spreading}), takes $O(\ell/\beta k)$ rounds.
Since every player knew the entire graph from the initialization phase and has learned about all edge-changes, all players also know the updated graph and thus can locally compute the new solution without further communication. 
Apart from the slow initialization time, we emphasize that this na\"ive approach requires each player to have $\Omega(m)$ bits of local memory, which is unrealistic when considering large data sets. 

We present a randomized algorithm that has a significantly faster initialization time of $O \lt( \frac{n}{\beta k}\log n \rt)$ rounds, nearly matching the bounds of Theorem~\ref{thm:mm_lb} up to a factor of $(k\cdot \log(\beta n))$ in the case of an oblivious adversary.
For an adaptive adversary, the gap between lower and upper bound is at most a $(\sqrt{\beta\,k}\cdot\log n)$-factor.

\newcommand{\algthm}{
Suppose that the adversary may add or delete up to $\ell$ edges per batch.
There exists a randomized dynamic algorithm for maximal matching that has an initialization time of $O\lt( \lceil\frac{n}{\beta\,k}\rceil \log n \rt)$ rounds with high probability (in $n$), %
and each player $P_i$ uses at most $O\lt(\max\set{n,\beta k,|G[V(P_i)]|}\cdot\log n\rt)$ bits of local memory. %
The update time $T$ is bounded as follows:
\begin{compactenum}
\item For the oblivious adversary: $T = O\lt( \lceil \frac{\ell}{\beta\,k} \rceil \log(\beta\, k) \rt)$ rounds w.h.p.\ (in $k$), and $T = O\lt( \lceil \frac{\ell}{\beta\,k} \rceil \log(\beta\,n) \rt)$ rounds w.h.p.\ (in $n$). %
\item For the adaptive adversary: $T = O\lt( \lceil \frac{\ell}{\sqrt{\beta\,k}} \rceil \log(\beta\,k) \rt)$ rounds w.h.p.\ (in $n$). %
\end{compactenum}
}
\begin{theorem} \label{thm:mm_random}
\algthm
\end{theorem}

In Section~\ref{sec:dmm-cc}, we show that our techniques lead to communication-efficient algorithms in the congested clique, in the sense that the message complexity is proportional to the number of updates.

\subsection{Related Work}
We start by describing prior work that considers dynamic updates in the $k$-machine model and the MPC model~\cite{karloff2010model}, as these are most closely related to our setting. 
While the $k$-machine model has been considered for a wide variety of problems such as PageRank approximation~\cite{pandurangan2021distributed} and graph clustering~\cite{bandyapadhyay2018near}, to the best of our knowledge, the only work studying dynamically-changing graphs in the $k$-machine model is the result of Gilbert and Lu~\cite{gilbert2020fast} on minimum spanning trees (MSTs). The primary technique used in their work is applying Euler Tours, which facilitates updating an MST in response to edge updates. 
The (static) maximal matching problem has not been studied in the $k$-machine model. 
While the approach of \cite{augustine2021efficient} provides a way to translate existing PRAM maximal matching algorithms to the $k$-machine model, their work crucially relies on a balanced \emph{edge partitioning}, and thus is not applicable to the $k$-clique message passing model, where we assume \emph{vertex partitioning}.  

The paper \cite{italiano2019dynamic} initiated dynamic problems in the MPC model and presented dynamic MPC algorithms for connectivity, minimum spanning tree and several matching problems. They also explored the relationship between classical dynamic algorithms and dynamic algorithms in the MPC model. Following the dynamic connectivity algorithms in the MPC model under a batch of updates \cite{dhulipala2020parallel},  Nowicki and Onak~\cite{nowicki2021dynamic} further studied dynamic MPC algorithms for minimum spanning forest, 2-edge connected components, and maximal matching with batch updates.

One may find that $k$-machine model and the MPC model share some similarities. We should notice that the major difference between the $k$-machine model and the MPC model is that, in the $k$-machine model, the total bandwidth is {$O(k^2 \log n)$} while the MPC model has a total bandwidth {$\tilde O(m)$} where $m$ is the number of edges of the input graph. In the MPC model, each machine can load all its input to other machines in one round, but in the $k$-machine model, this takes at least $\Omega(S/k)$ rounds where $S$ is an upper bound on the number of edges incident to the vertices hosted by each machine. In~\cite{nowicki2021dynamic}, their idea to deal with a batch of updates happening on a maximal matching is to reduce it to a maximal matching problem on a graph with vertex cover size at most $O(k)$, which can be solved by static algorithms of maximal matching in \cite{behnezhad2019exponentially}. 
Note that the total bandwidth in the $k$-clique message passing model is $O(k^2\log n)$ bits per round, and thus it is unclear how to efficiently run dynamic algorithms designed for the MPC model in our setting.

Since the pioneering work of
\cite{censor2016optimal}, several advancements have been made in the field of dynamic distributed models, e.g.,~\cite{bamberger2019local, censor2021fast, censor2021finding}.  Notably, \cite{censor2021fast} gave  dynamic algorithms for maximal matching in the {\congest} model using $O(1)$ amortized time complexity but the worst case time complexity is $O(n)$. 
Recently, \cite{foerster2021input} conducted  a study on dynamic problems in the {\congest} model and Congested Clique model. They demonstrated that for any problem, there exists a batch-dynamic congested clique algorithm using $O(\lceil \alpha/n \rceil)$ rounds and $O(m\log n)$ bits of auxiliary state when there are $\alpha$ edge label changes in a batch. Also~\cite{antaki2022near} gave distributed dynamic algorithms for some classical symmetry breaking problems and their techniques are majorly based on dynamic algorithms in the centralized settings.

\section{A Lower Bound for Batch-Dynamic Maximal Matching}\label{sec:dyn-mm}

In this section, we prove Theorem~\ref{thm:mm_lb}. 
We first present the proof for the oblivious adversary, which is technically more involved, due to the fact that the adversary neither controls the input partition nor has any knowledge of the current maximal matching when choosing the updates. 

For both cases, oblivious and adaptive adversaries, we consider the following lower bound graph $G$, formed by $q=n/3$ disjoint line segments $L_1,\dots,L_q$, and each $L_i$ consists of three nodes $x_{i,1}$, $x_{i,2}$, and $x_{i-3}$ that are connected by a path of length $2$; for simplicity, we assume that $n/3$ and $n/k$ are integers.
Figure~\ref{fig:lb1} gives an example of this construction.
If a vertex $u$ is part of the line segment $L_j$, we say that $j\in [q]$ is the \emph{index of $u$}.

\subsection{Oblivious Adversary} \label{sec:oblivious}

Initially, the algorithm computes some maximal matching on $G$. 
Thus, there exists a set $S$ of $q=n/3$ unmatched vertices in total, and $S$ contains exactly one unmatched vertex per line segment. %
For a player $P_i$, we define $U_i\subseteq [q]$ to be the set of corresponding line segment indices of the vertices in $S$ that are hosted by $P_i$.
Let $U_i' \subseteq U_i$ denote the indices of the unmatched vertices at $P_i$ such that, for every $j \in U_i'$, player $P_i$ hosts only the (single) unmatched vertex of line segment $L_j$, and neither of the other two matched vertices in $L_j$. 
We define the event {$\Good$ (``large'')} as
\begin{align}
	\Good = \bigwedge_{i=1}^{k} \lt( |U_i'| \ge |U_i| - \frac{12n}{k^{2}} \rt), \label{eq:good}
\end{align}
which captures the case when every set $U_i'$ is of sufficiently large size.
We now show that $\Good$ is very likely to occur.

\newcommand{\eventgood}{Event $\Good$ occurs with probability at least $1 - n^{-\Omega(1)}$.}
\begin{lemma} \label{lem:good}
 \eventgood
\end{lemma}

\begin{proof}
Consider some player $P_i$.
For every $j\in [q]$, define the indicator random variable $X_j^{(i)}$ to be $1$ if and only if $P_i$ hosts at least two vertices from line segment $L_j$, and let $X^{(i)}=\sum_{j=1}^{q}X_j^{(i)}$. 
Due the random vertex partitioning process, it follows that $\Pr\lt[ X_j^{(i)} \rt] = \binom{3}{2}\frac{1}{k^2} + \frac{1}{k^3} \le \frac{6}{k^2}$, , since $k = \Omega\lt( \log^3n \rt)$. 
Consequently, we have 
$\EE\lt[ X^{(i)} \rt] \le \frac{6n}{k^2}$.
As the line segments are disjoint, we know that $X_1^{(i)},\dots,X_q^{(i)}$ are independent.
Applying a standard Chernoff bound shows that $\Pr\lt[ X^{(i)}\ge \frac{12n}{k^2} \rt] \le 2^{-12n/k^2} \le \frac{1}{n^2}$, where we have used  that $k \le \sqrt{n}/\log n$, by the premise of theorem.
Taking a union bound over the $k$ players shows that 
\begin{align}
\bigwedge_{i=1}^{k} \lt( X^{(i)}\le \frac{12n}{k^{2}} \rt)  \label{eq:all_xs}
\end{align}
happens with high probability.
The lemma follows, because \eqref{eq:all_xs} implies that, for every player $P_i$,
\begin{align}
	|U_i'| \ge |U_i| - X^{(i)} \ge |U_i| - \frac{12n}{k^2}.\notag
\end{align}
\end{proof}

Clearly, there must exist a player $P_1$ that hosts a set of at least $q/k = n/3k$ unmatched vertices after computing the initial matching, i.e., $|U_1| \ge n/3k$, and this invariant remains true even when conditioning on event $\Good$. 
For the remainder of the proof of Theorem~\ref{thm:mm_lb}, we focus on the amount of information learned by $P_1$ during the course of the update procedure.
Define $\Good_{U_1'}$ to be the event that $U_1'$ is ``large'', formally,
\begin{align}
	|U_1'| \ge \frac{n}{3k} - \frac{12n}{k^2} \ge \frac{n}{6k}.
	\label{eq:good1}
\end{align}
Note that the second inequality holds for sufficiently large $n$, since $k=\omega(1)$ and $k=o(n)$ by assumption.
Clearly, $\Good_{U_1'}$ is implied by event $\Good$, and hence we immediately obtain the following from Lemma~\ref{lem:good}:
\begin{lemma} \label{lem:good1}
 Event $\Good_{U_1'}$ (see Ineq.~\eqref{eq:good1}) occurs with probability at least $1 - n^{-\Omega(1)}.$
\end{lemma}

\noindent\textbf{Adversarial Strategy.} The adversary samples a set of indices $I$, by independently including each $i\in [q]$ with probability $\frac{\ell}{n}$, i.e., $|I|=\frac{\ell\,q}{n}=\ell/3$ in expectation.
If the resulting set $I$ has a size greater than $\ell$, the adversary simply restarts the sampling process until $|I|\le \ell$. 
Then, for each $i\in I$, the adversary removes one of the two edges of segment $L_i$, chosen independently and uniformly at random.

Let $A \subseteq I$ be the set of \emph{affected indices}, which contains every index $i\in I$, for which the adversary deletes the matched edge $e$ of $L_i$, i.e., $e$ is \emph{not} incident to the unmatched vertex in $L_i$.
We define $A_1 = A \cap U_1'$, i.e., $A_1$ is the subset of indices $i$, such that:
\begin{compactenum}
\item the only vertex of $L_i$ hosted by player $P_1$ is the unmatched vertex $v \in L_i$, and 
\item the adversary has deleted the matched edge of $L_i$, which is not incident to $v$.
\end{compactenum}
Let $\mathsf{Bal}_{A_1}$ be the event that $|A_1|$ is ``balanced'' with respect to $U_1'$, which is true if 
\begin{align}
|A_1| \in \lt[ \frac{\ell}{4n}|U_1'|, \frac{4\ell}{n}|U_1'|\rt]. \label{eq:a1_bal}
\end{align}

\newcommand{\goodu}{For a given $U_1'$ and conditioned on $\Good_{U_1'}$, event $\mathsf{Bal}_{A_1}$ occurs, i.e., Eq.~\ref{eq:a1_bal} holds, with probability at least $1 - n^{-\Omega(1)}$.}
\begin{lemma} \label{lem:a1_bal}
\goodu
\end{lemma}

\begin{proof}
Suppose that $U_1'\!=\! u_1'$ and {$|u_1'|\ge n/(6k)$}, i.e., $\Good_{U_1'}$ holds.
Since the randomness used by the adversary for sampling $I$ and choosing the deleted edges (which determines $A$) is independent of the random vertex partitioning, it follows that 
\begin{align}
	\Pr\lt[ i \in A \ \md|\ i \in u_1',U_1'\!=\! u_1', \Good_{u_1} \rt] 
  = \Pr\lt[ i \in A \rt] 
  = \frac{1}{2}\Pr\lt[ i \in I \rt] 
  = \frac{\ell}{2n}, \label{eq:prob_a}
\end{align}
Consider the indicator random variable $Y_i = \mathbf{1}_{i \in A \cap u_1' }$ and note that $|A_1| = \sum_{i=1}^{q}Y_i$.
We have
\begin{align}
	\EE\lt[ |A_1|\ \md|\ U_1' \!=\! u_1', \Good_{U_1'} \rt] 
	&= \sum_{i=1}^{q} \Pr\lt[ Y_i \ \md|\ U_1' \!=\! u_1',\Good_{U_1'} \rt] \notag\\ 
	&= 
	\begin{multlined}[t]
  \sum_{i=1}^{q} \Pr\lt[ i \in A \ \md|\ i\in u_1',U_1' \!=\! u_1',\Good_{U_1'} \rt] \\
  		\cdot \Pr\lt[ i\in u_1' \ \md|\ U_1' \!=\! u_1',\Good_{U_1'} \rt]. 
  \end{multlined}
      \notag\\ 
	\intertext{Since $\Pr\lt[ i\in u_1' \ \md|\ U_1' \!=\! u_1',\Good_{U_1'} \rt]=0$ if $i\notin u_1'$, the right-hand side simplifies to}
	&= \sum_{i\in u_1'} \Pr\lt[ i \in A \ \md|\ i\in u_1',U_1' \!=\! u_1',\Good_{U_1'} \rt]  \notag\\ 
	\ann{by Eq.~\eqref{eq:prob_a}}
  &= \frac{\ell}{2n}\cdot|u_1'|. \label{eq:a_exp}
\end{align}
Note that, conditioned on events $U_1' \!=\! u_1'$ and $\Good_{U_1}$, variables $Y_1,\dots,Y_{q}$ remain independent.
By applying a standard Chernoff bound (see Theorem~4.4 in \cite{Mitzenmacher2005ProbabilityAC}), it follows that 
\begin{align}
	\Pr\lt[ |A_1|\ge \frac{\ell}{n}|u_1'|\ \md|\ U_1' \!=\! u_1',\Good_{U_1'} \rt] 
	&\le
	\exp \lt( -\frac{\ell}{2n}\frac{|u_1'|}{3} \rt)\notag\\ 
	\ann{since $\Good_{U_1}$ implies $|u_1'|\ge\frac{n}{6k}$}
	&\le
	\exp \lt( -\frac{\ell}{36k} \rt)\notag\\ 
	\ann{since $\ell\ge k\log n$}
	&\le
	n^{-\Omega(1)}.\notag
\end{align}
Similarly, we can bound the lower tail by Theorem~4.5 in \cite{Mitzenmacher2005ProbabilityAC}, which yields
\begin{align}
	\Pr\lt[ |A_1|\le \frac{\ell}{4n}|u_1'|\ \md|\ U_1' \!=\! u_1',\Good_{U_1'} \rt] 
	&\le
	\exp \lt( -\frac{\ell}{8n}\frac{|u_1'|}{2} \rt)\notag\\ 
	\ann{since $\Good_{U_1'}$ implies $|u_1'|\ge\frac{n}{6k}$}
	&\le
	\exp \lt( -\frac{\ell}{16k} \rt)\notag\\ 
	\ann{since $\ell\ge k\log n$}
	&\le
	n^{-\Omega(1)}.\notag
\end{align}
Therefore, $\mathsf{Bal}_{A_1}$ does not occur with probability at most $2n^{-\Omega(1)}\le n^{-\Omega(1)}$.
\end{proof}

\noindent\textbf{Notation.} We slightly abuse notation and write ``$P_i$'' to refer to the random variable that represents the local state of player $P_i$ after the adversary removes the edges but before any update-computation of the algorithm has taken place; we assume that the local state also includes the public random bits. 
We frequently compute the conditional entropy of a random variable $X$ with respect to a random variable $Y$ and some event $Z\!=\! z$. To improve readability, we use the notation $\HH\lt[ X \ \md|\ Y, z \rt]$ to mean $\HH\lt[ X\ \md|\ Y, Z\!=\! z \rt]$, i.e., the conditional entropy of $X$ with respect to random variable $Y$ conditioned on event $Z\!=\! z$.\footnote{
We provide some basic definitions from information theory in Appendix~\ref{sec:tools}.}

\begin{lemma} \label{lem:a1}
$\displaystyle
\HH\lt[ A_1 \ \md|\ {|A_1| \!=\! a,P_1\!=\! p_1, U_1' \!=\!u_1',\mathsf{Bal}_{A_1}},\Good_{U_1'} \rt]
=
\Omega\lt( \frac{\ell}{k}\cdot\log_2\lt( \frac{n}{\ell} \rt)  \rt)
$
\end{lemma}

\begin{proof}

Define $p(a_1) = \Pr\lt[ A_1 \!=\! a_1 \ \md|\ {a,p_1,u_1',\mathsf{Bal}_{A_1}},\Good_{U_1'} \rt]$. 
To obtain a bound on $p(a_1)$ that holds for any given $a_1$, we need to estimate the number of possible choices for $A_1$ under the given conditioning. 
Since $A_1\subseteq U_1'$ and $|A_1|\!=\! a$, there are $\binom{|u_1'|}{a}$ possible ways for choosing $A_1$. 
Recall that the adversary chooses the indices in $A$ (and hence also $A_1$) independently from the vertex partitioning, which tells us that $A_1$ has uniform probability over the $\binom{|u_1'|}{a}$ possible choices. 
We have 
\begin{align}
\log_2 \frac{1}{p(a_1)} 
	&\ge
\log_2\binom{|u_1'|}{a} \notag\\ 
	\ann{since $\binom{n}{k}\ge\lt(\frac{n}{k}\rt)^{k}$}
	&\ge a\cdot\log_2\lt( \frac{|u_1'|}{a} \rt)\notag\\ 
  \ann{since $\mathsf{Bal}_{A_1}$ implies $a\in[(\ell/4n)|u_1'|,(\ell/n)|u_1|']$}
	&\ge \frac{\ell}{4n}|u_1'|\cdot\log_2\lt( \frac{n}{\ell} \rt) \notag\\ 
	\ann{since $\Good_{U_1'}$ implies $|u_1'|\ge6n/k$}
	&= \Omega\lt( \frac{\ell}{k}\cdot\log_2\lt( \frac{n}{\ell} \rt)  \rt) \label{eq:H_a1}
\end{align}
It follows that
\begin{align}
\HH\lt[ A_1 \ \md|\ 
        {a,p_1,u_1',\mathsf{Bal}_{A_1},\Good_{U_1'}}
    \rt]
	&= \sum_{a_1}
  		p(a_1)
			\cdot
      \log_2 \frac{1}{p(a_1)}\notag\\ 
	\ann{by \eqref{eq:H_a1}}
	&\ge \Omega\lt( \frac{\ell}{k}\log\lt( \frac{n}{\ell} \rt) \rt)  
  		\sum_{a_1} p(a_1) %
	= \Omega\lt( \frac{\ell}{k}\log\lt( \frac{n}{\ell} \rt) \rt)  
  		\notag
\end{align}
\end{proof}
\enlargethispage{2\baselineskip}

\begin{figure}[t]
  \centering
\begin{subfigure}[t]{0.45\textwidth}  
  \centering
	\includegraphics[scale=0.70]{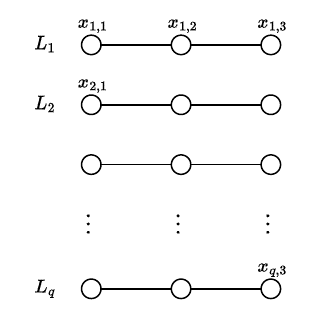}
  \caption{
		The input graph consists of $q=n/3$ line segments, each of length $2$. 
  }
  \label{fig:lb1}
\end{subfigure}
\hfill
\begin{subfigure}[t]{0.45\textwidth}
  \centering 
	{\includegraphics[scale=0.70]{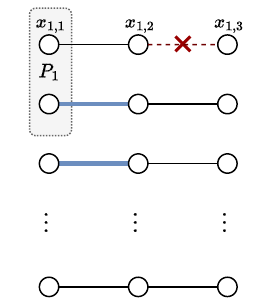}}
  \caption{
		The thick blue edges show the maximal matching after removing the (matched) edge $\set{x_{1,2},x_{1,3}}$.
		Player $P_1$ hosts $x_{1,1}$ and $x_{2,1}$.
  }
  \label{fig:lb2}
\end{subfigure}
\caption{The lower bound graph construction. 
    After the adversary deletes the matched edge $\set{x_{1,2}, x_{1,3}}$ in the graph shown in Figure~\ref{fig:lb2}, player $P_1$, who is unaware of this update, must learn about it from another player to know that it needs to output $\set{x_{1,1},x_{1,2}}$ as part of the matching.
}
\label{fig:lb}
\end{figure}

\vspace{-3em}
We now combine Lemmas~\ref{lem:good1}, \ref{lem:a1_bal}, and \ref{lem:a1} to bound the initial uncertainty about the set $A_1$ from the point of view of player $P_1$.
For an event $\mathcal{E}$, we use $\mathbf{1}_{\mathcal{E}}$ to denote the indicator random variable that is $1$ if and only if $\mathcal{E}$ occurs.
Since conditioning on a random variable cannot increase the entropy, we have
\begin{align}
	\HH\lt[ A_1 \ \md|\ P_1  \rt] 
	&\ge
    \HH\lt[ A_1 \ \md|\ 
        P_1, U_1', \mathbf{1}_{\Good_{U_1'}} 
    \rt] 	\notag\\ 
	&\ge
    \Pr\lt[ \Good_{U_1'} \rt] 
    \cdot
    \HH\lt[ A_1 \ \md|\ 
        P_1, U_1', \Good_{U_1'} 
    \rt] 	\notag\\ 
  \ann{by Lem.~\ref{lem:good1}}
	&\ge
    \tfrac{1}{2}
    \cdot
    \HH\lt[ A_1 \ \md|\ 
        P_1, U_1', \Good_{U_1'} 
    \rt] 	\notag\\ 
	&=
    \tfrac{1}{2}
    \cdot
    \sum_{u_1} 
    \Pr\lt[ U_1 \!=\! u_1'\ \md|\ \Good_{U_1'} \rt] 
		\cdot
    \HH\lt[ A_1 \ \md|\ 
        P_1, {U_1'\!=\! u_1',\Good_{U_1'}}
    \rt] 	\notag\\ 
	&\ge
    \tfrac{1}{2}
    \sum_{u_1'}%
    \Pr\lt[ u_1'\ \md|\ \Good_{U_1'} \rt] 
		\cdot
    \HH\lt[ A_1 \ \md|\ 
        P_1, \mathbf{1}_{\mathsf{Bal}_{A_1}},{u_1'}, \Good_{U_1'}
    \rt] 	\notag\\ 
	&\ge
		\begin{multlined}[t]
    \tfrac{1}{2}
    \sum_{u_1'}%
		\Big(
    \Pr\lt[ u_1'\ \md|\ \Good_{U_1'} \rt] 
		\cdot
		\Pr\lt[ \mathsf{Bal}_{A_1}\ \md|\ u_1', \Good_{U_1'} \rt] \\
		\phantom{---------}
		\cdot
    \HH\lt[ A_1 \ \md|\ 
        P_1, {\mathsf{Bal}_{A_1},u_1', \Good_{U_1'}}
    \rt] 	
		\Big)
    \end{multlined}
    \notag\\ 
	\ann{by Lem.~\ref{lem:a1_bal}}
	&\ge
		\begin{multlined}[t]
    \tfrac{1}{2}
    \cdot
    \tfrac{1}{2}
    \sum_{u_1'}%
    \Pr\lt[ u_1'\ \md|\ \Good_{U_1'} \rt] 
		\cdot
    \HH\lt[ A_1 \ \md|\ 
        P_1, {\mathsf{Bal}_{A_1},u_1', \Good_{U_1'}}
    \rt] 
    \end{multlined}
    \notag\\ 
	&\ge
		\begin{multlined}[t]
    \tfrac{1}{4}
    \sum_{u_1'}%
    \Pr\lt[ u_1'\ \md|\ \Good_{U_1'} \rt] 
		\cdot
    \HH\lt[ A_1 \ \md|\ 
        P_1, |A_1|, {\mathsf{Bal}_{A_1},u_1',\Good_{U_1'}}
    \rt] 
    \end{multlined}
    \notag\\ 
	&=
		\begin{multlined}[t]
		\tfrac{1}{4}
    \sum_{u_1'}%
    \Pr\lt[ u_1' \ \md|\ \Good_{U_1'} \rt] 
		\cdot
		\sum_{a,p_1} 
		\Big(
		\Pr\lt[ |A_1|=a,P_1\!=\! p_1\ \md|\ \mathsf{Bal}_{A_1},u_1',\Good_{U_1'} \rt] 
		\\
		\cdot
    \HH\lt[ A_1 \ \md|\ 
        {a,p_1,u_1',\mathsf{Bal}_{A_1},\Good_{U_1'}}
    \rt] 	
		\Big)
    \end{multlined} \notag\\ 
	\ann{by Lem.~\ref{lem:a1}}
	&=
		\begin{multlined}[t]
 	  \Omega\lt( \frac{\ell}{k}\cdot\log_2\lt( \frac{n}{\ell} \rt)  \rt)
    \sum_{u_1'}%
    \Pr\lt[ u_1'\ \md|\ \Good_{U_1'} \rt] 
		\cdot
		\sum_{a,p_1} 
		\Pr\lt[ a, p_1\ \md|\ \mathsf{Bal}_{A_1},u_1',\Good_{U_1'} \rt] 
    \end{multlined} \notag\\ 
   &= 
 	  \Omega\lt( \frac{\ell}{k}\cdot\log_2\lt( \frac{n}{\ell} \rt)  \rt), \label{eq:mm_lb1}
\end{align}
where the final step follows because both of the sums over the probabilities evaluate to $1$.

Let random variable $\mathsf{Out}_{1}$ be the edges that are output by $P$ as part of the maximal matching after the update is complete, and let $\mathbf{1}_{\text{Corr}}$ be the indicator random variable that the algorithm succeeds, which equals $1$ with high probability according to the premise of the theorem. 
\begin{lemma} \label{lem:out1}
$\displaystyle
  \HH \lt[ A_1\ \md|\ \mathsf{Out}_{1}, P_1 \rt] 
	= O(1).
$
\end{lemma}
\begin{proof}
By the chain rule of conditional entropy, we have 
\begin{align}
  \HH \lt[ A_1\ \md|\ \mathsf{Out}_{1}, P_1 \rt] 
  &\le 
  \HH \lt[ A_1, \mathbf{1}_{\text{Corr}} \md|\ \mathsf{Out}_{1}, P_1 \rt] \notag\\
  &=
  \begin{multlined}[t]
  \HH \lt[ A_1\ \md|\  \mathbf{1}_{\text{Corr}} , \mathsf{Out}_{1}, P_1 \rt]  
  +
  \HH \lt[ \mathbf{1}_{\text{Corr}} \md|\ \mathsf{Out}_{1}, P_1 \rt] 
  \end{multlined}
     \notag\\
  \ann{since $\HH \lt[ \mathbf{1}_{\text{Corr}} \md|\ \mathsf{Out}_{1}, P_1 \rt]\le 1$}
  &\le
  \HH \lt[ A_1\ \md|\  \mathbf{1}_{\text{Corr}} , \mathsf{Out}_{1}, P_1 \rt] 
  +
  1\notag
  \\
  &=
  \begin{multlined}[t]
		\Pr\lt[ \mathbf{1}_{\text{Corr}} \!=\! 1 \rt]
    \HH \lt[ A_1\ \md|\  \mathsf{Out}_{1}, P_1, \mathbf{1}_{\text{Corr}} \!=\! 1\rt] \\
    +
    \Pr\lt[ \mathbf{1}_{\text{Corr}} \!=\! 0 \rt] 
    \HH \lt[ A_1\ \md|\  \mathsf{Out}_{1}, P_1, \mathbf{1}_{\text{Corr}} \!=\! 0\rt] 
    + 1. 
  \end{multlined} \notag\\ 
  &\le
  \begin{multlined}[t]
    \HH \lt[ A_1\ \md|\  \mathsf{Out}_{1}, P_1, \mathbf{1}_{\text{Corr}} \!=\! 1\rt] \\
    +
    \Pr\lt[ \mathbf{1}_{\text{Corr}} \!=\! 0 \rt] 
    \HH \lt[ A_1\ \md|\  \mathbf{1}_{\text{Corr}} \!=\! 0\rt] 
    + 1. 
  \end{multlined} \notag\\ 
	\ann{since $\Pr\lt[  \mathbf{1}_{\text{Corr}} \!=\! 0\rt] \le 1/n$}
  &\le
  \begin{multlined}[t]
    \HH \lt[ A_1\ \md|\  \mathsf{Out}_{1}, P_1, \mathbf{1}_{\text{Corr}} \!=\! 1\rt] 
    + \frac{1}{n}
    \HH \lt[ A_1\ \md|\  \mathbf{1}_{\text{Corr}} \!=\! 0\rt] 
    + 1. 
  \end{multlined} \notag\\ 
  &\le  
    \HH \lt[ A_1\ \md|\  \mathsf{Out}_{1}, P_1, \mathbf{1}_{\text{Corr}} \!=\! 1\rt] 
    + O(1), 
\label{eq:mm_lb2}
\end{align}
where the final step follows since $\HH \lt[ A_1\ \md|\ \mathbf{1}_{\text{Corr}} \!=\! 0 \rt] = O(n)$.

Conditioned on the event $\mathbf{1}_{\text{Corr}} \!=\! 1$, player $P_1$ must output the edge incident to every previously unmatched vertex $v\in L_j$, for all $j \in A_1$; see Figure~\ref{fig:lb2} on page~\pageref{fig:lb2} for an example.
This means that $P_1$ must have learned the entire set of indices in $A_1$ by the time the algorithm completes its update,
which implies that
\begin{align}
  \HH \lt[ A_1\ \md|\  \mathsf{Out}_{1}, P_1, \mathbf{1}_{\text{Corr}} \!=\! 1\rt]
  = 0, \label{eq:mm_lb3}
\end{align}
and completes the proof of Lemma~\ref{lem:out1}.
\end{proof}

Let $\Pi$ be the transcript of messages received by player $P_1$ during the update.
Note that $\mathsf{Out}_{1}$ is a function of $\Pi$ and the local state of $P_1$, which includes the public randomness.
By the data-processing inequality (see Lemma~\ref{lem:dataprocessing}), it follows that
\begin{align}
  \II \lt[\,\Pi : A_1\ \md|\ P_1 \rt] 
  &\ge 
  \II \lt[\, \mathsf{Out}_{1} : A_1\ \md|\ P_1 \rt]\notag \\
  &=
    \HH \lt[ A_1\ \md|\ P_1 \rt]
    - 
    \HH \lt[ A_1\ \md|\ \mathsf{Out}_{1}, P_1 \rt] \notag\\ 
  \ann{by Ineq.~\eqref{eq:mm_lb1} and Lem.~\ref{lem:out1}}
  &\ge
		\Omega\lt( \frac{\ell}{k}\log \lt( \frac{n}{\ell} \rt) \rt)
    - 
    O(1)\notag\\ 
  &= 
  \Omega\lt( \frac{\ell}{k} \log k \rt), \label{eq:mut_inf8}
\end{align}
where the final step follows since $\ell \le n/2k$, according to the premise of the theorem.
In each round, player $P_1$ can receive at most $c \beta\log n$ bits over every one of its $k-1$ communication links, for some constant $c\ge 1$. 
Thus there are $2^{c \beta\log n}+1\le 2^{c \beta\log n+1}$ possible messages (including the empty message) that $P_1$ may receive from a distinct player $P_j$ in a single round. 
Let $\mathcal{M}_r$ be the concatenation of the messages that $P_1$ receives in a given round $r$ from the $k-1$ other players. 
By the above derivation, there are at most $2^{c\beta(k-1)\log n +1}$ possible values for $\mathcal{M}_r$.
Let $T$ denote the worst case update time of player $P_1$, and note that $\Pi = \mathcal{M}_1\dots \mathcal{M}_T$. 
It follows that there are at most $2^{c \beta(k-1)T\log n +1}$ possible values for $\Pi$.
The entropy $\HH\lt[ \Pi \rt]$ is maximized when $\Pi$ is uniformly distributed, i.e.,  $\HH\lt[ \Pi \rt] \le \log_2 2^{c \beta(k-1)T\log n +1}$, and \eqref{eq:mut_inf8} implies that
\begin{align*}
{c\beta(k-1)T\log n +1}
 \ge \HH\lt[ \Pi \rt] 
      \ge \II \lt[\, \Pi : A_1\ \md|\ P_1 \rt]
      = \Omega\lt( \frac{\ell}{k} \log k \rt).
\end{align*}
This shows that $T= \Omega\lt( \frac{\ell\log k}{\beta\,k^2\log n} \rt)$, as required.

\subsection{Extending the Proof to the Adaptive Adversary} \label{sec:omitted_lb}

We now consider the {adaptive adversary}, who not only determines the initial (balanced) vertex partitioning, but is also aware of the current state of the nodes, which, in particular includes the computed maximal matching.
As before, we focus on the graph $G$ that consists of a collection of $q$ disjoint line segments $L_1,\dots,L_{q}$, where each $L_i = (x_{i,1},x_{i,2},x_{i,3})$ forms a path of length $2$ on vertices $x_{i,1}$, $x_{i,2}$, and $x_{i,3}$; see Figure~\ref{fig:lb1} on page~\pageref{fig:lb1}. 
To simplify the presentation, we assume that $\frac{q}{n/k} = \frac{k}{3}$ is an integer. 
Conceptually, we view the graph as an $(q\times 3)$-size checkerboard $B$ where $x_{i,j}$ is located at coordinate $(i,j)$. 
It is possible to \emph{tile} $B$, i.e., cover all spaces without any overlaps by using $k$ dominos $D_1,\dots, D_k$, each of size $(n/k)\times 1$, see, e.g., \cite{rosen}.
The set of coordinates of $B$ covered by $D_i$ represents the vertices assigned to player $P_i$, which results in a balanced vertex partition, where each player obtains $n/k$ vertices.
Figure~\ref{fig:lb3} gives an example of such a tiling. 

\begin{figure}[t] %
  \centering
  \includegraphics[scale=0.85]{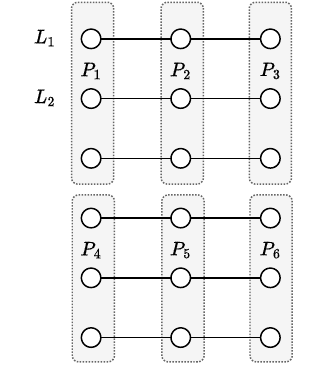}
\caption{The initial vertex partition (``tiling'') of the input graph chosen by the adaptive adversary. Each shaded area corresponds to the vertices hosted by one player.}
\label{fig:lb3}
\end{figure}

This ensures the following property:

\begin{fact} \label{fact1}
Every player hosts at most $1$ vertex per line segment. 
\end{fact}

Observe that any maximal matching must leave exactly one vertex per line segment unmatched.
Similarly to the oblivious case, there must exist a player $P$ who hosts a set $S_P$ of at least $n/3k$ unmatched vertices.  
A crucial difference, however,  is that the adaptive adversary knows $P$ and $S_P$ when selecting the edge updates.
We define $I_P$ to be the set of indices of the line segments that have a vertex in $S_P$, and note that Fact~\ref{fact1} implies that $|I_P|\ge n/3k$.

Then, the adversary samples a set $A_P \subseteq I_P$ of size $\ell$ uniformly at random, and deletes the matched edge $e_i\in L_i$ from the graph, for each line segment $L_i$ with an index $i\in A_P$.
As player $P$ hosts no vertex incident to $e_i$, it is unaware of these deletions unless it communicates with the other players.
Observe that there are ${|I_P| \choose \ell}\ge {n/3k \choose \ell}$ choices for $A_P$, all of which are equally likely. 
It follows that
\begin{align*}
  \HH\lt[ A_P \mid P  \rt] 
  &\ge 
  \log_2 {n/3k \choose \ell} \notag\\ 
  &\ge 
  \log_2 \lt( \frac{n/3k}{\ell} \rt)^{\ell}\notag\\ 
  &= \ell\cdot\log_2\frac{n}{3k\,\ell}\notag\\ 
	\ann{since $\ell\le n/2k$}
  &= \Omega\lt( \ell \rt).
\end{align*}
Similarly to the analysis for the oblivious adversary we obtain that, conditioned on the output of $P$, the entropy of $A_P$ becomes close to $0$. 
Thus it follows that player $P$ needs to learn at least $\Omega(\ell)$ bits of information while executing the update algorithm, which takes $\Omega\lt( \frac{\ell}{\beta\,k\log n} \rt)$ rounds.

\section{An Algorithm for Batch-Dynamic Maximal Matching}\label{sec:dyn-mm-ran}

In this section, we describe an algorithm that satisfies the bounds claimed in Theorem~\ref{thm:mm_random}, and prove its correctness.

We first describe some information dissemination tools that our algorithm makes frequent use of. 
In particular, the next lemma provides a technique for efficiently spreading a set of messages to all players.

\begin{lemma}[$\spreading$] \label{lem:spreading}
Suppose that there are $N \ge \beta(k-1)$ tokens, each of size $\Theta(\log n)$ bits, located at arbitrary players. There exists a deterministic algorithm $\spreading$ such that each player can receive all $N$ tokens in $O\lt(\frac{N}{\beta\,k}\rt)$ rounds.
\end{lemma}

\begin{proof}
    We describe the $\spreading$ algorithm and argue the claimed time complexity bound.
Suppose player $P_i$ holds a set of $x_i$ tokens initially, and recall that $\sum_{i=1}^{k}x_i = N$. 
\begin{compactenum}
\item \textbf{Redistribution:} $P_i$ splits this set into batches of size exactly $\beta(k-1)$ (with the possible exception of the last batch). Then, in parallel for each $j \in [k-1]$, player $P_i$ sends the $j$-th token of the first batch directly to the player, to which it is connected via its $(j \bmod (k-1))$-th link. %
Afterwards, it proceeds with the second batch and so forth. 
All players perform this in parallel until all of their batches have been processed, which takes $O(N/\beta\,k)$ rounds. 

At this point, each player $P_i$ has $x_i' \le \beta(k-1)-1$ tokens left that it has not yet redistributed. 
Next, the players exchange the counts $x_1',\dots,x_k'$ which takes $O(1)$ rounds. %
Every player $P_a$ fixes some arbitrary order  $m^a_1,\dots,m^a_{x_a'}$ of its tokens, and, for every $j \in [x_a']$, locally computes the \emph{rank of token $m^a_j$}, denoted by  $r(m^a_j)$, as 
\[
r(m^a_j) = \sum_{i=1}^{a-1} x_i' + j.
\]
It is straightforward to verify that this results in a unique rank for all tokens of all players.
Then, for every $r \in [x_a']$, player $P_a$ sends the token with rank $r$ to player $P_{(r \bmod k) + 1}$; if $(r\bmod k) +1=a$, it simply keeps the token in its local memory.
Since $x_a' \le \beta(k-1)-1$ and the tokens are almost equally distributed over all incident communication links, these remaining tokens can be sent in parallel in just a single round.
This ensures that every player $P_i$ holds a set $T_i$ of at most $O(N/k)$ tokens locally.
\item \textbf{Dissemination:} 
$P_i$ sequentially broadcasts every token in its set $T_i$ to all players.
After $O(N/\beta\,k)$ rounds, every player is guaranteed to have received every token.
\end{compactenum}

\end{proof}

\subsection{Initialization} \label{sec:mm_random_init}

The initialization phase of our randomized algorithm simply computes a maximal matching on the input graph. 
As a maximal matching in a graph $G$ is equivalent to a maximal independent set (MIS) in its line graph $\mathcal{L}_G$, a common approach is to simulate an MIS algorithm on $\mathcal{L}_G$.
However, in contrast to the congested clique model, the overhead of sending $O(m)$ messages for simulating $\mathcal{L}_G$ is too expensive in our setting, as the total available bandwidth in a single round is only $O(\beta\,k^2\log n)$ bits. 
Instead, we employ the maximal matching algorithm of Israeli and Ittai~\cite{israeli1986fast}.
As we also leverage this algorithm for handling the updates under an oblivious adversary, we describe its implementation in the $k$-clique message passing model in more detail.

Every player keeps track of a Boolean variable $\textsf{matched}_u$, initialized to $0$, for each of its hosted vertices $u$, which indicates if an edge incident to $u$ is included in the matching.
We now describe a single iteration of the algorithm of \cite{israeli1986fast}:
\begin{compactenum}
\item \textbf{Edge Sampling Step}: Every player $P$ uniformly at random chooses an incident edge $e=\set{u,v}$, for each of its vertices $u$, and sends a $\msg{marked,\text{$e$}}$ message to player $P(v)$ via $\spreading$.
Conceptually, we consider the marked edge $e' = \set{u,v'}$ to be a directed edge $e' = (u,v')$ (pointing from $u$ to $v'$)  and use $\vec{G}$ to denote the directed graph consisting of the sampled edges of all players.
After this step, every player $P$ knows the subgraph of $\vec{G}$ induced by the marked edges that are incident to its own vertices.

\item \textbf{Reduce In-degrees Step}: Player $P$ chooses a random incoming edge $e'=(v', u) \in E(\vec{G})$ for each of its vertices $u$ (if any), and sends a $\msg{selected,\text{$e'$}}$ message to player $P(v')$ via $\spreading$.
Let $\bar{G}$ be the \emph{undirected} graph that only consists of all edges $e$, for which a $\msg{marked,\text{$e$}}$ as well as a $\msg{selected,\text{$e$}}$ message was sent. 

\item \textbf{Match-up Step:}
Next, each player $P$ samples, for each of its vertices $u$ in $\bar{G}$, an incident edge $e''=\set{ u,v''} \in E(\bar{G})$ uniformly at random, and sends $\msg{request,\text{$e''$}}$ to $P(v'')$.
If $P$ also receives a $\msg{request,\text{$e''$}}$ message from $P(v'')$, then it sets $\textsf{matched}_{u}\leftarrow1$ (and $P(v'')$ will set $\textsf{matched}_{v''}\leftarrow1$ in turn). %

\item \textbf{Pruning Step}: 
Consider an edge $e=\set{u,v}$ that was matched in the previous step, and suppose that the ID of $u$ is smaller than the ID of $v$.
We say that player $P(u)$ is \emph{responsible} for $e$.
Finally, we need to compute the residual graph that does not include any edges to already matched nodes. 
To this end, each player $P$ creates a message $\msg{matched,\text{$e$}}$, for every edge $e$ for which it is responsible, and this message is sent to every other player via $\spreading$,.
After these exchanges, the players have learned about all edges that have at least one matched endpoint, which they discard from the graph. 
The players locally check whether there are any edges remaining and, if so, proceed to the next iteration. 
The maximal matching is simply the union over the matchings obtained in the individual phases. 
\end{compactenum}
The next lemma quantifies the performance of this algorithm in the $k$-clique message passing model:

\newcommand{\lemIttai}{
Consider an $n'$-node graph $G'$ in which the vertices are partitioned among $k$ players, either randomly or in a balanced manner. %
Then, there exists a maximal matching algorithm that terminates in 
$O\lt( \lt\lceil \frac{n'}{\beta\,k} \rt\rceil \log n' \rt)$ rounds with high probability (in $n'$).
}
\begin{lemma} \label{lem:ittai}
\lemIttai
\end{lemma}
\begin{proof}
For the case of random vertex partitioning, observe that each player $P$ hosts $O(n'/k)$ vertices in expectation and, by a standard Chernoff bound, $P$ has $O(\frac{n'}{k}  \log n')$ vertices with high probability, which means that (w.h.p.) we have a balanced vertex partitioning. 
We condition on this event in the remainder of the proof.

The correctness of the computed matching follows from \cite{israeli1986fast}.
Therefore, we focus on proving the bound on the time complexity.
In the \textbf{edge sampling}, \textbf{reduce in-degrees}, and \textbf{match-up} steps, each player $P$ sends at most one message for each vertex in $V(P)$, and may need to receive one message for each of the incident edges of the vertices in $V(P)$. 
It follows that in total there are at most $n'$ $O(\log n)$-bit size messages that need to be sent in each of these steps, which takes $O(n'/\beta\,k)$ rounds via $\spreading$ (see Lemma~\ref{lem:spreading}).

Finally, we consider the \textbf{pruning step}:
Let $S$ be the set of edges matched in the current iteration, and note that $|S|\le n'/2$.
Since the players send these edges by invoking $\spreading$ (see Lemma~\ref{lem:spreading}), every player learns the entire set $S$ within $O\lt( \lt\lceil \frac{|S|}{\beta\,k}  \rt\rceil \rt) = O\lt(\lt\lceil  \frac{n'}{\beta\,k} \rt\rceil \rt)$ rounds. %

In \cite{israeli1986fast}, they show that 
$\Theta(\log (|E(G'|))) = \Theta(\log n')$ iterations are sufficient in expectation for computing a maximal matching.
In more detail, they show that each iteration removes a $\gamma$-fraction of the edges in expectation, for some constant $\gamma > 0$. 
Thus, by considering $c\cdot\log |E(G')|$ iterations, for a sufficiently large constant $c$, it follows that the expected number of edges in the residual graph is at most $O \lt( \frac{m'}{\gamma^{c\log |E(G')|}} \rt) = O \lt( \frac{1}{|E(G')|^{c}} \rt)$.
By Markov's inequality, we obtain a maximal matching in $O(\log |E(G')|) = O(\log n')$ iterations with high probability.
\end{proof}

Our algorithm ensures that the following invariant holds after the initialization phase, as well as every time the algorithm has finished computing a new matching that was triggered by a batch of edge updates:   
 
\begin{invariant} \label{inv:random}
Every player knows, for each of its hosted vertices $v$, the matched edge incident to $v$ (if any), and, for each neighbor $u$ of $v$, it knows whether $u$ is matched. 
\end{invariant}

By instantiating Lemma~\ref{lem:ittai} with the initial input assignment, i.e., $n'=n$, we immediately obtain the following:
\begin{lemma} \label{lem:mm_init}
The initialization time is $O\lt( \lceil\frac{n}{\beta\,k}\rceil\log n \rt)$ rounds with high probability and Invariant~\ref{inv:random} holds.
\end{lemma}

Recall that a batch of updates may consist of $\ell_1$ edge deletions and $\ell_2$ edge insertions, where $\ell_1 + \ell_2 \le \ell$.
Conceptually, our algorithm treats this as two separate update steps with batches of size $\ell_1$ and $\ell_2$, respectively.
This allows us to assume that each batch contains only one type of edge-update in the following two subsections.

\subsection{Edge Insertions} \label{sec:mm_random_edge_insertions}

\newcommand{\lemedgeadd}{
If the (oblivious or adaptive) adversary inserts a batch of $\ell$ edges, then the update time is $O(\lceil \ell/\beta k\rceil)$ rounds and Invariant~\ref{inv:random} holds.
}
\begin{lemma} \label{lem:edge_additions}
\lemedgeadd
\end{lemma}

\begin{proof}
If an edge is inserted between vertices $u$ and $v$, and one of them is already matched, there is nothing to do. 
Thus, assume that all edges are inserted between unmatched nodes.
Every player that hosts a vertex incident to some of the inserted edges sends the IDs of the endpoints of these edges to player $P_1$.
Since this amounts to $O(\ell\log n)$ bits in total, $P_1$ can receive all of these messages in $O(\ell/\beta k)$ rounds. 
Then, $P_1$ locally computes a maximal matching $M'$ on the subgraph induced by the edges that were inserted between unmatched nodes and sends $M'$ to every other player. Since $|M'|=O(\ell)$, this again takes $O(\ell/\beta k)$ rounds. 
From this information, each player learns about all newly matched edges that are incident to its own vertices, which ensures that Invariant~\ref{inv:random} is satisfied.
\end{proof}

\subsection{Edge Deletions}\label{subsec:edge-deletion}
 
We now focus on the significantly more difficult case of handling deletions of $\ell$ edges.
Since our algorithm can simply ignore deletions of non-matched edges, we assume that all removed edges were part of the current matching.  

Conceptually, we split these updates into $\lt\lceil\ell/\Gamma\rt\rceil$ smaller \emph{mini-batches} and process one mini-batch at a time.
All mini-batches have a size of exactly $\Gamma$, with the possible exception of the last one, which contains at most $\Gamma$ updates.
For the oblivious adversary, we choose a mini-batch size of $\Gamma=\beta\,k$, whereas, for the adaptive adversary, we set $\Gamma=\lfloor \beta\,\sqrt{k} \rfloor$. 

For a given subset of vertices $S \subseteq V(G)$, let $G[S]$ denote the subgraph induced by $S$.
Suppose that we are given some (not necessarily maximal) matching $M$ in $G$.
We say that a vertex is \emph{free} or \emph{unmatched}, if it does not have an incident edge in $M$.
We define the \emph{free degree $\free_{G,M}(u)$} of a node $u$ to be the number of neighbors incident to $u$ in $G$ that are free according to the matching $M$.
For a given matching $M'$ in $G[S]$, we sometimes abuse notation and simply write $\free_{S,M'}(u)$ instead of $\free_{G[S],M'}(u)$; we omit the subscripts when they are clear from the context.

The next lemma captures important properties of the resulting graph obtained right after the adversary has deleted a set of edges, denoted by $D$. 
We define $V_f$ to be the set of unmatched nodes incident to edges in $D$, and let $V' \subseteq V(G) \setminus V_f$ be the set of unmatched vertices that have neighbors in $V_f$. 
Figure~\ref{fig:random_fig} gives an example of these sets. 

\begin{figure}[t]
  \centering
\begin{subfigure}[t]{0.40\textwidth}
  \centering
  \includegraphics[]{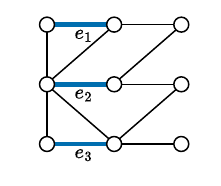}
\caption{A maximal matching $\set{e_1,e_2,e_3}$ in graph $G$.}
\label{fig:random_fig1}
\end{subfigure}
\hspace{5mm}
\begin{subfigure}[t]{0.5\textwidth} 
  \centering
  \includegraphics[]{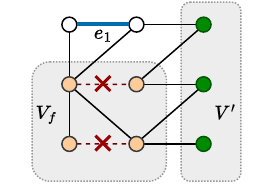}
  \caption{The adversary has deleted edges $e_2$ and $e_3$ from $G$. The light orange-shaded nodes form $V_f$, whereas $V'$ consists of the dark green-shaded nodes.}
  \label{fig:random_fig2}
\end{subfigure} %
\caption{}
\label{fig:random_fig}
\end{figure}

\newcommand{\lemmmvertices}{Suppose that Invariant~\ref{inv:random} holds in the current graph $G$ with maximal matching $M$. Assume that the adversary deletes a set $D$ of at most $\Gamma$ edges, and let $G'$ denote the resulting graph. 
Then:
\begin{compactenum}
\item[(a)] $|V_f| \le 2\Gamma$;
\item[(b)]for all $u \in V'$ we have $\free_{G',M}(u) \le 2\Gamma$;
\item[(c)]$V'$ is an independent set in $G$ (i.e., $G[V']$ contains no edges). 
\end{compactenum}}

\begin{lemma}\label{lem:mm_vertices}
    \lemmmvertices
\end{lemma}
\begin{proof} 
Property~(a) is immediate from the fact that at most $\Gamma$ edges are deleted. 
Next, consider Property~(c): 
Recall that every $v\in V'$ was itself unmatched in $G$. 
Since Invariant~\ref{inv:random} ensures that the algorithm has computed a maximal matching on $G$, there cannot be any edge between nodes in $V'$. 

Finally, for Property~(b), observe that (c) implies $\free_{G,M}(v)=0$ for all $u\in V'$, and the only way that the free degree of $v$ can increase in $G'$ is if $u$ is connected to some of the nodes in $V_f$.
\end{proof}

For simplicity, we assume that the given mini-batch of updates results in the deletion of a set $D$ of exactly $\Gamma$ edges from $G$ that were part of the matching $M$, yielding the graph $G'$.

\subsubsection{Special Case: Updating Graphs with High Free Degree in $O(1)$ Rounds} 
Before we explain the details of our general update algorithm in Section~\ref{sec:highlevel}, we first give a simple and time-optimal way of handling graphs where every node has a high free degree.

As a first step, every player sends the IDs of all its nodes that are in $V_f$ to every other player via $\spreading$, and hence every player learns the IDs of all nodes of $V_f$.
Let $L=(v_1,v_2,\dots,v_{2\Gamma})$ be an ordering of the vertices in $V_f$ in increasing order of their IDs. 
By locally processing the vertices in $V_f$ sequentially according to $L$, each player $P$ tries to pair up every node $v_i \in V_f$ with some edge $e_i$ that connects $v_i$ to one of its unmatched neighbors $w \in V' \cap V(P)$.
If $P$ pairs up a free vertex $w$ with $v_i$, then it marks $w$ as taken, and does not consider it when pairing up the remaining nodes $v_{i+1},\dots,v_{2\Gamma}$.
After processing the entire list $L$, player $P$ sends, for each $i\in [1,\Gamma]$, the paired-up edges $\set{v_i,w}$ and $\set{v_{i+\Gamma},w'}$ directly to $P_j$, where $j = (i \bmod k) + 1$; %
if $P$ was unable to find a suitable edge for a vertex, it simply omits the corresponding message.

Consequently, every player $P_j$ receives at most $2$ edges from each of the $k-1$ other players, for each vertex $v_i$ such that $j = (i \bmod k)+1$. 
Subsequently, $P_j$ performs a local ``filtering'' step by discarding from the received edges all but one incident edge $e_i$ for $v_i$ and an incident edge $e_{i+\Gamma}$ for $v_{i+\Gamma}$. 
If $e_i$ exists, then it becomes part of the matching, and we handle $e_{i+\Gamma}$ analogously. %
Then, $P_j$ broadcasts the (at most two) matched edges to all other players, and each player updates the incident matched edges of its vertices accordingly.

\newcommand{\lemhighdeg}{
Let $G$ be a dynamic graph where, for every vertex $u$, it holds that $f_{G,M}(u) \geq 2(k\Gamma + \Gamma)+1$. There exists an deterministic algorithm that can update the maximal matching in constant rounds. 
}

\begin{lemma} \label{lem:high_deg}
\lemhighdeg

\end{lemma}

\begin{proof}
We first show that every node $u$ with $f_{G,M} (u) \geq 2(k\Gamma+\Gamma)+1$ is matched:  
Each time that a player $P$ pairs-up a node $v_i \in V_f$, which occurs at index $i$ of the ordered list $L$, with one of its neighbors $w \in V(P)\cap V'$, this may reduce the available edges for pairing up a node $v_j$ ($j \in [i+1,2\Gamma]$) by at most one. 
Note that $w$ itself does not occur in $L$.
According to the algorithm, there are $ k$ players that are executing this pair-up process in parallel for $v_i$, and thus there are at most $2k\Gamma$ nodes in $V'$ that are used for pairing up. 
Since $v_j$ has at most $2\Gamma$ neighbors in $V_f$, it follows that the number of unpaired free neighbors of $v_j$ in $V'$ is at least $\free_{G',M}(v_j) - 2k\Gamma - 2\Gamma$.
If $v_j$ has high free degree, then $\free_{G',M}(v_j) - 2k\Gamma - 2\Gamma \ge 1$, and hence there is at least one edge $e$ available for being paired up with $v_j$ at any point in this process, which will be sent to the corresponding player $P_j$, where $j = (i \bmod k)+1$.
It follows that any high free degree node is guaranteed to be matched and hence has an incident edge in $A_1$.

Next, we argue that $M_1$ forms a valid matching on $G'$: 
Observe that all edges used for pairing-up have one endpoint in $V_f$ and the other in $V'$.
Each vertex $v \in V'$ is used at most once during this process by the player who hosts $v$.
Together with the fact that when the corresponding player $P_j$ discards all but one edge for $v_i$ and $v_{i+\Gamma}$, this ensures that no edges in $A_1$ share an endpoint.
According to the algorithm, every player broadcasts its matched edges to everyone else and thus all players learn $M_1$.

Finally, we show the claimed bound on the running time:
The initial spreading of the $2\Gamma$ IDs in $V_f$ requires $O(\Gamma/\beta k) = O(1)$ rounds, according to Lemma~\ref{lem:spreading}.
For the process of pairing up vertices with edges, observe that, for each player $P_j$, there are at most $2\beta$ vertices $v_i$ in the list $L$, where 
$j = (i \bmod k)+1$, since $|L| \le 2\Gamma \le 2\beta k$.
Thus, $P_j$ receives at most $2\beta$ messages from every other player, which means that $P_i$ can learn about all paired-up edges for all vertices $v_{i}$ and $v_{i+\Gamma}$ ($j = (i \bmod k)+1$) within $O(1)$ rounds.
Subsequently, $P_j$ broadcasts at a message for each of the (at most) $2\beta$ vertices to every other player, which again takes $O(1)$ rounds. 
\end{proof}

\subsubsection{High-level Overview} \label{sec:highlevel}
In our analysis, we will show that the algorithm maintains Invariant~\ref{inv:random} at the end of its update step. 
We conceptually split the update algorithm into two phases. 
In Phase~1, we focus on the unmatched nodes in $V_f$ that have at least $2\Gamma+1$ free neighbors in $V'$.
We sample a free neighbor for each of these nodes and show that this process matches a constant fraction of these nodes with high probability in each iteration, by making use of the fact that the corresponding random variables are negatively correlated. 
Once the remaining number of these unmatched nodes is sufficiently small (i.e., $O(\log n)$), we sample a small subgraph of their free edges that we quickly disseminate to every player, enabling us to locally compute a matching for these nodes. 
In Phase~2, our approach depends on the type of adversary:
For the oblivious adversary, we leverage the fact that the vertices are randomly partitioned among the players, and argue that we can directly use the same matching algorithm as in the initialization step.
On the other hand, for the adaptive adversary, we use the fact that the size of the mini-batch is only $O(\sqrt{\beta k})$, which, together with the degree reduction achieved in Phase~1, allows us to aggregate the entire subgraph induced by the remaining unmatched nodes in $V_f$ at a single player.
In the remainder of this section, we give a detailed description of these two phases.

\subsubsection{Phase~1} 
Let $T \subseteq V_f$ be the subset of nodes that each have at least $2\Gamma+1$ unmatched neighbors in $V'$  { with respect to $M_1$ where $M_1 = M\setminus D$.}
In this phase, our goal is to match all vertices in $T$.
Recall that, at the start of Phase~1, each player $P$ learned the IDs and name of the respective player $P(u)$ of every vertex $u \in V_f$, and hence it also knows this information for every vertex in $T$.

Let $T_P$ be the subset of $T$ hosted by player $P$.
Observe that $P$ knows $T_P$ since $V_f$ and $M_1$ are known to everyone.
We use the shorthand $\free_{V'}(u) = \free_{G'[\{u\} \cup V'],M_1}(u)$ for the number of unmatched neighbors of $u$ that are in $V'$, and we define $\free_{V'}^{(i)}(u)$ to be the number of such neighbors hosted at player $P_i$, i.e., $\free_{V'}^{(i)}(u) =  \free_{G'[\{u\} \cup (V' \cap V(P_i))],M_1}(u)$. 
Note that $P_i$ can locally compute $\free_{V'}(u)$ and $\free_{V'}^{(j)}(u)$ for every $j\in [k]$ and every vertex $u\in V(P_i)$ from its knowledge of $M_1$ and $V_f$.

We perform a sequence of $\Theta(\log k)$ iterations:
Our goal, in a given iteration, is to ensure that a node in $T$ has a constant probability of being matched to a node that is not in $T$.
Note that $\Theta(\log(\beta k))$ iterations suffice, because Lemma~\ref{lem:mm_vertices} tells us that the size of $T$ is at most $O(\Gamma)= O(\beta k)$.

All players execute the following process in parallel:
For each $u \in T_P$, player $P$ first samples an index $i$ from $[k]$ according to the probability distribution 
$\lt( 
	\frac{\free_{V'}^{(1)}(u)}{\free_{V'}(u)},\dots, \frac{\free_{V'}^{(k)}(u)}{\free_{V'}(u)}
\rt)$.
Then, $P$ sends a $\msg{match-up!}$ request for $u$ to player $P_i$, who in turn samples an unmatched neighbor $v$ uniformly at random from the unmatched neighbors of $u$ in $V(P_i)\cap V'$, and considers the  $\msg{match-up!}$ message as being ``received by $v$''. 
If $P_i$ does not receive any other $\msg{match-up!}$ message for $v$ in this iteration, then it adds the edge $\set{u,v}$ to the matching and also informs $P$ about this.
Note that all of these messages are sent via $\spreading$ (see Lemma~\ref{lem:spreading}).

After processing all $\msg{match-up!}$ requests, each player $P$ updates its set $T_P$ by removing all nodes for which it found a matching in this iteration, as well as those nodes whose number of free neighbors in $V'$ has dropped below $2\Gamma+1$.
To fulfill these requirements, we need to ensure that a player knows when the free degree of one of its vertices changes, and thus we instruct each player to broadcast its newly matched nodes to everyone via $\spreading$ (see Lemma~\ref{lem:spreading}). 
Then, every player $P$ sends the new value of $|T_P|$ to $P_1$, who locally computes the sum, i.e., the updated size $|T|$, and broadcasts this count to everyone.
As long as $|T| \ge c_1\log n$, for a suitable constant $c_1>0$, the players move on to the next iteration by repeating the above process. 

\newcommand{\lemmmphasetwoone}{If $|T| \ge c_1 \log n$ at the start of an iteration, then a constant fraction of the nodes in $T$ are matched with high probability.}

\begin{lemma} \label{lem:mm_phase2_1}
\lemmmphasetwoone
\end{lemma}

\begin{proof}
We start our analysis by considering the probability of finding a matching for a node $u \in T_P$ in a given iteration.
According to the algorithm, a player $P$ randomly samples some player $P_i$ with probability $\frac{\free_{V'}^{(i)}(u)}{\free_{V'}(u)}$, who in turn uniformly samples a hosted free neighbor of $u$ with probability $\frac{1}{\free_{V'}^{(i)}(u)}$.
Since 
\[\frac{\free_{V'}^{(i)}(u)}{\free_{V'}(u)}\cdot \frac{1}{\free_{V'}^{(i)}(u)} = \frac{1}{\free_{V'}(u)},
\]
  it follows that the $\msg{match-up!}$ request for $u$ is indeed sent to a uniformly at random sampled node among its free neighbors outside $T$.
For each $v_j \in T$, let $X_j$ be the indicator random variable that is $1$ if and only if $v_j$'s $\msg{match-up!}$ request reaches the same node as the request of another node in this iteration, i.e., if $X_j=0$, then $v_j$ is matched.
It follows that 
\begin{align}
	\Pr\lt[ X_j \!=\! 0 \rt] 
  &\ge \lt( 1 - \frac{1}{\free_{V'}(u)} \rt)^{|T|} \notag\\ 
	\ann{since $T \le 2\Gamma$}
	&\ge \lt( 1 - \frac{1}{\free_{V'}(u)} \rt)^{2\Gamma} \notag\\ 
	\ann{since $\forall u \in T\colon \free_{V'}(u) \ge 2\Gamma+1$}
	&\ge \lt( 1 - \frac{1}{2\Gamma+1} \rt)^{2\Gamma} \notag\\ 
  \ann{since $1 - x \ge e^{-2x}$ for $x<\tfrac{1}{2}$}
	&\ge e^{-4\Gamma/(2\Gamma+1)}.\notag
\end{align}
Conversely, this tells us that the probability that $u$ is not matched in this iteration is at most some constant 
\begin{align}
\delta \le 1 - e^{-3} \le 1 - e^{-4\Gamma/(2\Gamma+1)}.  \label{eq:mm_delta}
\end{align}

A technical complication stems from the fact that $X_i$ and $X_j$ are not necessarily independent, for nodes $v_i$ and $v_j$ that share neighbors.
In particular, if a node $v_i$ is not matched (i.e., $X_i \!=\! 1$), then its request reached the same free node $w \notin T$  that was also reached by the request of some other node. 
Conditioning on this event may skew the probability distribution of $X_j$ for some distinct node $v_j$.
However, since this event fixes the destination of two request messages to the same vertex, this cannot \emph{increase} the probability of $X_j \!=\! 1$.
In other words, conditioning on this event, it may be less likely that another $\msg{match-up!}$ message collides with $v_j$'s own $\msg{match-up!}$ request, possibly lowering the probability of event $X_j \!=\! 1$, and hence it holds that $\Pr\lt[ X_j \!=\! 1 \ \md|\ X_i \!=\! 1 \rt] \le \Pr\lt[ X_j \!=\! 1 \rt]$. 
A similar argument applies when considering a subset of vertices $\set{v_1,\dots,v_s}$, that is,
\begin{align}
 \Pr\lt[ \bigwedge_{j = 1}^{s} (X_j \!=\! 1) \rt] 
 =
 \prod_{j=1}^{s}
 \Pr\lt[ X_j \!=\! 1 \ \md|\ \bigwedge_{r=1}^{j-1} (X_r \!=\! 1) \rt] 
 \le 
 \prod_{j=1}^{s}
\Pr\lt[ X_j \!=\! 1 \ \rt] 
 \le 
 \delta^s \label{eq:negative}
\end{align}
The bound \eqref{eq:negative} enables us to instantiate the generalization of the Chernoff Bound stated in Lemma~\ref{lem:chernoff_generalized}, where $N = |T| \ge c_1 \log n$, $\eta = 1 - e^{-4}$, and we recall that $\delta \le 1 - e^{-3}$ due to \eqref{eq:mm_delta}:
\begin{lemma}[see Theorem~1.1 in \cite{impagliazzo2010constructive}] \label{lem:chernoff_generalized}
    Let $X_1,\dots,X_N$ be (not necessarily independent) Boolean random variables and suppose that, for some $\delta \in [0,1]$, it holds that, for every index set $S \subseteq [N]$, $\Pr\lt[\bigwedge_{i \in S} X_i\rt] \le \delta^{|S|}$.
    Then, for any $\eta \in [\delta,1]$, we have $\Pr\lt[\sum_{i=1}^N X_i \ge \eta N\rt] \le e^{- 2N(\eta - \delta)^2}$.
  \end{lemma}

It follows that
$
\Pr\Big[ \sum_{j=1}^{|T|}X_j \ge \lt( 1 - e^{-e} \rt) |T| \Big] \le e^{-\Omega(\log n)} \le n^{-\Omega(1)},
$
which completes the proof of Lemma~\ref{lem:mm_phase2_1}.
\end{proof}

Once we have reached an iteration after which it holds that $|T| < c_1\log n$, we sample a subgraph $S$ as follows:
Each player samples uniformly at random, for each of its vertices in $T$, a subset of $c_2\log^2n$ incident edges connecting to free neighbors, where $c_2>0$ is a suitable constant.
The players perform an all-to-all exchange of the sampled edges using $\spreading$, after which everyone knows the entire subgraph $S$. 
In the proof of the following Lemma~\ref{lem:mm_phase2_2}, we show that this enables the players to locally compute a matching that is maximal for all remaining nodes in $T$ by using some fixed deterministic function $g$ that maps each possible subgraph $S$ to a set of matched edges. 
Note that $g$ is hard-coded into the algorithm and hence known to all players.

\newcommand{\lemmmphasetwwo}{
If $|T| < c_1\log n$ at the end of an iteration, then all remaining nodes in $T$ are matched in $O(1)$ additional rounds.
}

\begin{lemma}\label{lem:mm_phase2_2}
    \lemmmphasetwwo
\end{lemma}

\begin{proof}
We first argue the bound on the time complexity.
By assumption, any given player holds at most $O(\log n)$ nodes from $T$.
According to the algorithm, we select $c_2\log^2 n$ edges to unmatched neighbors for each of these nodes. 
Note that this is possible because every node in $T$ has a free degree of at least $2\Gamma+1$ and due to the assumptions that $k = \Omega(\log^4n)$, which means that, for both, oblivious and adaptive adversary, we have 
\[
2\Gamma+1 = \Omega(\sqrt{\beta k}) = \Omega\lt( \log^2 n \rt).
\]
Moreover, since $|T|=O(\log n)$, the total number of messages that need to be sent is $O(\log^3n)=O(k)$.
Thus, Lemma~\ref{lem:spreading} tells us that we can apply $\spreading$ to disseminate these edges to all other players in $O(1)$ rounds.

After $\spreading$ has completed, all players know the entire sampled subgraph $S$. 
Observe that each node in $T$ has a degree of $\Omega(\log^2 n)$ in $S$, whereas $|T| = O(\log n)$.
Thus, it follows that there always exists a maximal matching in $S$ that matches every node in $T$. 
\end{proof}

The next lemma combines Lemmas~\ref{lem:mm_phase2_1} and \ref{lem:mm_phase2_2} to show a significant reduction of the free degree of the remaining unmatched nodes in $V_f$.

\begin{lemma} \label{lem:mm_phase2}
At the end of Phase~1, every player knows the matching $M_2= M_1 \cup N_2$, where $N_2$ denotes the set of newly matched edges, and $M_1 = M \setminus D$ .
Phase~1 takes $O(\log(\beta k))$ rounds with high probability, and matches every vertex in $v\in V_f$ that had a free degree of at least $4\Gamma+1$, i.e., if $\free_{G',M_1}(v)\ge4\Gamma+1$.
\end{lemma}
\begin{proof}
Recall that $T$ denotes the subset of nodes in $V_f$ that have at least $2\Gamma+1$ free neighbors in $V'$.
It follows that any node in $V_f$ with a free degree of at least $4\Gamma+1$ must have $4\Gamma+1-2\Gamma$ free neighbors in $V'$, and hence will be in set $T$. 
Lemma~\ref{lem:mm_phase2_1} tells us that $|T| \le c_1\log n$ (w.h.p.) after $O(\log(\beta k))$ iterations, which, together with Lemma~\ref{lem:mm_phase2_2}, implies the sought time complexity bound.

By the description of the algorithm, all players exchange the newly matched edges after each iteration.
Moreover, once $|T| < c_1 \log n$, every player uses the same function $g$ to compute the matching on the sampled subgraph $S$, which ensures that all players know $N_2$.
Since every player knows $M_1$ after Phase~1, it follows that everyone also knows $M_2$.
\end{proof}

\subsubsection{Phase~2} 
Let $L \subseteq V_f$ be the remaining unmatched vertices in $V_f$ after Phase~2.
In the final phase, we take care of the vertices in $L$, 
which we call the \emph{low free degree vertices}, as Lemma~\ref{lem:mm_phase2} ensures that each of them has a free degree of at most $4\Gamma$.
We provide two approaches for this phase, depending on the strength of the adversary.

\medskip\noindent\textbf{Oblivious Adversary.}
We first consider the setting where the adversary must choose all updates before the vertices are randomly partitioned among the players. 
Let $G_L \subseteq G[V_f \cup V']$ be the subgraph induced by the free edges incident to nodes in $L$.
Observe that $G_L$ contains all remaining unmatched nodes that have free neighbors. 
Similarly to the initialization phase, we use the algorithm of \cite{israeli1986fast}, for which we have described an implementation in   Section~\ref{sec:mm_random_init}. 
According to Lemma~\ref{lem:mm_phase2}, every $u\in L$ has at most $4\Gamma+1$ unmatched neighbors when considering the matching $M_2$, and we know from Lemma~\ref{lem:mm_vertices}(a) that $|L| \le |V_f| \le 2\Gamma$.
It follows that $G_L$ has at most $\Gamma=O(\Gamma)$ vertices. 
Since the adversary is oblivious to the random vertex partitioning, we can instantiate Lemma~\ref{lem:ittai} with $n'= \Gamma = O(\beta k)$. 
This immediately implies the following lemma:

\begin{lemma} \label{lem:mm_phase3_1}
Phase~2 ensures Invariant~\ref{inv:random} and takes $O(\log(\beta k))$ rounds with high probability (in $k$) against an oblivious adversary, assuming a mini-batch size of $\Gamma = \beta k$ edge deletions.
\end{lemma}

\medskip\noindent\textbf{Adaptive Adversary.}
When the adversary has full knowledge of the vertex partitioning among the players, it may happen that all nodes in $L$ are distributed among a small number of players, thus rendering the vast majority of available resources (i.e., players and bandwidth between them) useless.
To avoid this pitfall, we leverage the assumption that the mini-batch size is only $\lfloor \sqrt{\beta k} \rfloor$.
In the proof of Lemma~\ref{lem:mm_phase3_2}, we show that we can aggregate the entire subgraph $G_L$ at a single player $P_1$, who locally computes a matching, and subsequently informs all other players of the result.

\newcommand{\lemmmphasethreetwo}{Phase~2 ensures Invariant~\ref{inv:random} and takes $O(1)$ rounds against an adaptive adversary, assuming a mini-batch size of $\Gamma= \lfloor \sqrt{k} \rfloor$.}

\begin{lemma} \label{lem:mm_phase3_2}
\lemmmphasethreetwo
\end{lemma}
\begin{proof}
The correctness of the computed matching is immediate from the fact that $P_1$ locally computes the solution.
Since each node in $L$ has a free degree of at most $4\Gamma\le4\sqrt{\beta k}$ and $|L|\le2\Gamma\le2\sqrt{\beta k}$, it follows that $G_L$ consists of $O(\beta k)$ edges. 
Thus, according to Lemma~\ref{lem:spreading}, player $P_1$ can receive the entire graph in just $O(1)$ rounds by executing $\spreading$. 
\end{proof}

\subsection{Completing the Proof of Theorem~\ref{thm:mm_random}}
The bound on the initialization time follows from Lemma~\ref{lem:mm_init}, which also tells us that Invariant~\ref{inv:random} holds after the initialization step.

For edge insertions, Lemma~\ref{lem:edge_additions} ensures that the update time for handling $\ell$ edge insertions is only $O\lt( \frac{\ell}{\beta k} \rt)$ (for both oblivious and adaptive adversaries). 
Next, we derive the bound on the update time for edge deletions:
Recall that we have $\ell$ updates that we split into mini-batches of size $\Gamma$.
Lemma \ref{lem:mm_phase2} shows that Phases~1 requires $O \lt( \log(\beta k) \rt)$ rounds with high probability.
Finally, Lemma~\ref{lem:mm_phase3_1} shows that under an oblivious adversary, Phase~2 takes $O(\log(\beta k))$ rounds with high probability (in $k$), for mini-batches of size $\beta k$, which yields a total update time of $O\lt( \frac{\ell}{\Gamma} \log k\rt)$.
On the other hand, Lemma~\ref{lem:mm_phase3_2} shows that, under an adaptive adversary, Phase~2 can be done in $O \lt( \log k \rt)$ rounds for mini-batch updates of size $\Theta(\sqrt{k})$, and hence the three phases take at most $O\lt( \frac{\ell}{\Gamma} \log(\beta k)\rt)$ rounds. 
The claimed bounds of Theorem~\ref{thm:mm_random} follow by plugging in $\Gamma=\beta k$ for the oblivious adversary and $\Gamma=\lfloor\sqrt{\beta k}\rfloor$ for the adaptive adversary, respectively.
 
To see that each player $P_i$ uses at most $O\lt(\max\set{n,|G[V(P_i)]|}\cdot\log n\rt)$ bits of local memory, where $|G[V(P_i)]|$ is the number of edges incident to vertices hosted by $P_i$, observe that each player sends and receives at most $1$ message per vertex in every step of the algorithm in Phases~1 and 2, which requires storing at most $O(n\log n)$ bits.
Moreover, each player needs to locally maintain at most $O(\log n)$ bits for each edge incident to its hosted vertices, using at most $O(|G[V(P_i)]|\log n)$ bits.

Finally, note that Lemma~\ref{lem:mm_phase3_1} and Lemma~\ref{lem:mm_phase3_2} guarantee that Invariant~\ref{inv:random} holds at the end of each update step, which ensures the correctness of the computed maximal matching.

\section{Batch-Dynamic Maximal Matching in the Congested Clique Model}\label{sec:dmm-cc}

It is not surprising that our algorithm also work in the congested clique due to the equivalence between the two models when $k=n$ and if each player gets exactly one vertex (see Section~\ref{sec:intro}). 
However, we now show that we can even get dynamic algorithms in the Congested Clique model that handle updates \emph{communication-efficiently}, which means that the number of messages sent grow proportionally to the number of edge-changes. 
To this end, we define the \emph{update message complexity} as the worst case number of messages sent following a batch of $\ell$ edge updates, assuming that the state of the nodes reflects a maximal matching prior to these updates. 
In Appendix~\ref{sec:omitted_cc}, we prove the following result:

\newcommand{\thmcc}{
  Consider the congested clique and suppose the adversary may add or delete up to $\ell$ edges per batch. There exists a randomized dynamic algorithm for maximal matching that, with high probability, has an initialization time of $O(\log \log \Delta)$ rounds. With high probability, the update message complexity is $O(\ell n)$, and the following hold:
\begin{compactenum}
\item The update time complexity is $O\lt (\lceil \frac{\ell}{n} \rceil \log n \rt )$ against an oblivious adversary.
\item The update time complexity is $O\lt (\lceil \frac{\ell}{\sqrt{n}} \rceil \log n \rt )$ against an adaptive adversary.
\end{compactenum}
}
\begin{theorem}\label{thm:cc}
\thmcc
\end{theorem}

We emphasize that the message complexity of the initialization algorithm in Theorem~\ref{thm:cc} can be as large as $O(n^2)$. 
In fact, obtaining a message complexity of $o(n^2)$ for maximal matching in the congested clique is itself an interesting open problem to the best of our knowledge.

\bibliography{ref}

\begin{thebibliography}{10}

\bibitem{acar2019parallel}
Umut~A Acar, Daniel Anderson, Guy~E Blelloch, and Laxman Dhulipala.
\newblock Parallel batch-dynamic graph connectivity.
\newblock In {\em The 31st ACM Symposium on Parallelism in Algorithms and
  Architectures}, pages 381--392, 2019.

\bibitem{acar2011parallelism}
Umut~A Acar, Andrew Cotter, Beno{\^\i}t Hudson, and Duru T{\"u}rkoglu.
\newblock Parallelism in dynamic well-spaced point sets.
\newblock In {\em Proceedings of the twenty-third annual ACM symposium on
  Parallelism in algorithms and architectures}, pages 33--42, 2011.

\bibitem{anderson2023parallel}
Daniel Anderson.
\newblock {\em Parallel Batch-Dynamic Algorithms Dynamic Trees, Graphs, and
  Self-Adjusting Computation}.
\newblock PhD thesis, Carnegie Mellon University, 2023.

\bibitem{antaki2022near}
Shiri Antaki, Quanquan~C Liu, and Shay Solomon.
\newblock Near-optimal distributed implementations of dynamic algorithms for
  symmetry breaking problems.
\newblock In {\em 13th Innovations in Theoretical Computer Science Conference
  (ITCS 2022)}. Schloss Dagstuhl-Leibniz-Zentrum f{\"u}r Informatik, 2022.

\bibitem{augustine2021efficient}
John Augustine, Kishore Kothapalli, and Gopal Pandurangan.
\newblock Efficient distributed algorithms in the k-machine model via pram
  simulations.
\newblock In {\em 2021 IEEE International Parallel and Distributed Processing
  Symposium (IPDPS)}, pages 223--232. IEEE, 2021.

\bibitem{bamberger2019local}
Philipp Bamberger, Fabian Kuhn, and Yannic Maus.
\newblock Local distributed algorithms in highly dynamic networks.
\newblock In {\em 2019 IEEE International Parallel and Distributed Processing
  Symposium (IPDPS)}, pages 33--42. IEEE, 2019.

\bibitem{bandyapadhyay2018near}
Sayan Bandyapadhyay, Tanmay Inamdar, Shreyas Pai, and Sriram~V Pemmaraju.
\newblock Near-optimal clustering in the k-machine model.
\newblock In {\em Proceedings of the 19th International Conference on
  Distributed Computing and Networking}, pages 1--10, 2018.

\bibitem{behnezhad2019exponentially}
Soheil Behnezhad, Mohammad~Taghi Hajiaghayi, and David~G Harris.
\newblock Exponentially faster massively parallel maximal matching.
\newblock In {\em 2019 IEEE 60th Annual Symposium on Foundations of Computer
  Science (FOCS)}, pages 1637--1649. IEEE, 2019.

\bibitem{censor2021fast}
Keren Censor-Hillel, Neta Dafni, Victor~I Kolobov, Ami Paz, and Gregory
  Schwartzman.
\newblock Fast deterministic algorithms for highly-dynamic networks.
\newblock In {\em 24th International Conference on Principles of Distributed
  Systems (OPODIS 2020)}. Schloss Dagstuhl-Leibniz-Zentrum f{\"u}r Informatik,
  2021.

\bibitem{censor2016optimal}
Keren Censor-Hillel, Elad Haramaty, and Zohar Karnin.
\newblock Optimal dynamic distributed mis.
\newblock In {\em Proceedings of the 2016 ACM Symposium on Principles of
  Distributed Computing}, pages 217--226, 2016.

\bibitem{censor2021finding}
Keren Censor-Hillel, Victor~I Kolobov, and Gregory Schwartzman.
\newblock Finding subgraphs in highly dynamic networks.
\newblock In {\em Proceedings of the 33rd ACM Symposium on Parallelism in
  Algorithms and Architectures}, pages 140--150, 2021.

\bibitem{facebookgiraph}
Avery Ching, Sergey Edunov, Maja Kabiljo, Dionysios Logothetis, and Sambavi
  Muthukrishnan.
\newblock One trillion edges: Graph processing at facebook-scale.
\newblock {\em {PVLDB}}, 8(12):1804--1815, 2015.
\newblock URL: \url{http://www.vldb.org/pvldb/vol8/p1804-ching.pdf}.

\bibitem{cover1999elements}
Thomas~M Cover.
\newblock {\em Elements of information theory}.
\newblock John Wiley \& Sons, 1999.

\bibitem{dhulipala2020parallel}
Laxman Dhulipala, David Durfee, Janardhan Kulkarni, Richard Peng, Saurabh
  Sawlani, and Xiaorui Sun.
\newblock Parallel batch-dynamic graphs: Algorithms and lower bounds.
\newblock In {\em Proceedings of the Fourteenth Annual ACM-SIAM Symposium on
  Discrete Algorithms}, pages 1300--1319. SIAM, 2020.

\bibitem{foerster2021input}
Klaus-Tycho Foerster, Janne~H Korhonen, Ami Paz, Joel Rybicki, and Stefan
  Schmid.
\newblock Input-dynamic distributed algorithms for communication networks.
\newblock {\em Proceedings of the ACM on Measurement and Analysis of Computing
  Systems}, 5(1):1--33, 2021.

\bibitem{gilbert2020fast}
Seth Gilbert and Lawrence Er~Lu Li.
\newblock How fast can you update your {MST}?
\newblock In Christian Scheideler and Michael Spear, editors, {\em {SPAA} '20:
  32nd {ACM} Symposium on Parallelism in Algorithms and Architectures, Virtual
  Event, USA, July 15-17, 2020}, pages 531--533. {ACM}, 2020.
\newblock \href {https://doi.org/10.1145/3350755.3400240}
  {\path{doi:10.1145/3350755.3400240}}.

\bibitem{graphx}
Joseph~E. Gonzalez, Reynold~S. Xin, Ankur Dave, Daniel Crankshaw, Michael~J.
  Franklin, and Ion Stoica.
\newblock Graphx: Graph processing in a distributed dataflow framework.
\newblock In {\em {USENIX} {OSDI} 2014}, pages 599--613, 2014.
\newblock URL:
  \url{https://www.usenix.org/conference/osdi14/technical-sessions/presentation/gonzalez}.

\bibitem{hourani2013distributed}
Khalid Hourani, Hartmut Klauck, William~K Moses~Jr, Danupon Nanongkai, Gopal
  Pandurangan, Peter Robinson, and Michele Scquizzato.
\newblock Distributed algorithms for large-scale graphs.
\newblock {\em arXiv e-prints}, pages arXiv--1311, 2023.

\bibitem{impagliazzo2010constructive}
Russell Impagliazzo and Valentine Kabanets.
\newblock Constructive proofs of concentration bounds.
\newblock In {\em Approximation, Randomization, and Combinatorial Optimization.
  Algorithms and Techniques: 13th International Workshop, APPROX 2010, and 14th
  International Workshop, RANDOM 2010, Barcelona, Spain, September 1-3, 2010.
  Proceedings}, pages 617--631. Springer, 2010.

\bibitem{israeli1986fast}
Amos Israeli and Alon Itai.
\newblock A fast and simple randomized parallel algorithm for maximal matching.
\newblock {\em Information Processing Letters}, 22(2):77--80, 1986.

\bibitem{italiano2019dynamic}
Giuseppe~F Italiano, Silvio Lattanzi, Vahab~S Mirrokni, and Nikos Parotsidis.
\newblock Dynamic algorithms for the massively parallel computation model.
\newblock In {\em The 31st ACM Symposium on Parallelism in Algorithms and
  Architectures}, pages 49--58, 2019.

\bibitem{karloff2010model}
Howard Karloff, Siddharth Suri, and Sergei Vassilvitskii.
\newblock A model of computation for mapreduce.
\newblock In {\em Proceedings of the twenty-first annual ACM-SIAM symposium on
  Discrete Algorithms}, pages 938--948. SIAM, 2010.

\bibitem{klauck2014distributed}
Hartmut Klauck, Danupon Nanongkai, Gopal Pandurangan, and Peter Robinson.
\newblock Distributed computation of large-scale graph problems.
\newblock In {\em Proceedings of the twenty-sixth annual ACM-SIAM symposium on
  Discrete algorithms}, pages 391--410. SIAM, 2014.

\bibitem{lenzen2013optimal}
Christoph Lenzen.
\newblock Optimal deterministic routing and sorting on the congested clique.
\newblock In {\em Proceedings of the 2013 ACM symposium on Principles of
  distributed computing}, pages 42--50, 2013.

\bibitem{lotker}
Zvi Lotker, Boaz Patt{-}Shamir, Elan Pavlov, and David Peleg.
\newblock Minimum-weight spanning tree construction in \emph{O}(log log
  \emph{n}) communication rounds.
\newblock {\em {SIAM} J. Comput.}, 35(1):120--131, 2005.
\newblock URL: \url{http://dx.doi.org/10.1137/S0097539704441848}, \href
  {https://doi.org/10.1137/S0097539704441848}
  {\path{doi:10.1137/S0097539704441848}}.

\bibitem{pregel}
Grzegorz Malewicz, Matthew~H. Austern, Aart J.~C. Bik, James~C. Dehnert, Ilan
  Horn, Naty Leiser, and Grzegorz Czajkowski.
\newblock Pregel: a system for large-scale graph processing.
\newblock In {\em SIGMOD Conference}, pages 135--146, 2010.

\bibitem{Mitzenmacher2005ProbabilityAC}
Michael Mitzenmacher and Eli Upfal.
\newblock {\em Probability and Computing: Randomized Algorithms and
  Probabilistic Analysis}.
\newblock Cambridge University Press, 2004.

\bibitem{nowicki2021dynamic}
Krzysztof Nowicki and Krzysztof Onak.
\newblock Dynamic graph algorithms with batch updates in the massively parallel
  computation model.
\newblock In {\em Proceedings of the 2021 ACM-SIAM Symposium on Discrete
  Algorithms (SODA)}, pages 2939--2958. SIAM, 2021.

\bibitem{pandurangan2021distributed}
Gopal Pandurangan, Peter Robinson, and Michele Scquizzato.
\newblock On the distributed complexity of large-scale graph computations.
\newblock {\em ACM Transactions on Parallel Computing (TOPC)}, 8(2):1--28,
  2021.

\bibitem{rosen}
Kenneth~H Rosen.
\newblock {\em Discrete mathematics and its applications}.
\newblock The McGraw Hill Companies, 2007.

\end{thebibliography}

\appendix

\section{Omitted Part from~Section~\ref{sec:dmm-cc} } \label{sec:omitted_cc}

For the initialization phase, we need to compute a maximal matching on the input graph. Our method is to transform the existing maximal matching algorithm of \cite{behnezhad2019exponentially} from the MPC model to the congested clique:

\begin{theorem}[\cite{behnezhad2019exponentially}]\label{thm:mpc-mm}
    Given an $n$-vertex graph $G$ with $m$ edges and maximum degree $\Delta$, there exists a randomized MPC algorithm for computing a maximal matching that takes $O(\log \log \Delta)$ rounds using $O(n)$ space per machine. The algorithm succeeds with probability $1-e^{-n^{\Omega(1)}}$ and requires an optimal total space $O(m)$.
\end{theorem}

By Theorem~\ref{thm:mpc-mm}, we can obtain a maximal matching in the Congested Clique model.
\begin{lemma}\label{lem:cc-mm}
    Given an $n$-vertex graph $G$ with $m$ edges and max degree $\Delta$, there exists a randomized algorithm in the Congested Clique model for computing a maximal matching that takes $O(\log \log \Delta)$ rounds with probability $1-e^{-n^{\Omega(1)}}$.
\end{lemma}
\begin{proof}
    By Theorem~\ref{thm:mpc-mm}, there is a randomized algorithm that can find a maximal matching in MPC model with $O(n)$ space per machine within $O(\log \log \Delta)$ rounds with probability $1-e^{-n^{\Omega(1)}}$. Now, we simulate that algorithm in the Congested Clique model as follows: 
    We fix some ordering of the nodes, and then select the first $O(\frac{m}{n})$ nodes as \emph{worker nodes}, which will perform the same functionality as the $O(\frac{m}{n})$ machines in the MPC model. 
Recall that in the Congested Clique model, there is a sorting algorithm using constant rounds \cite{lenzen2013optimal}. All nodes execute the sorting algorithm to obtain a sorted ordering of the edges. Then, for a suitable constant $c_1>0$, the nodes send the $i$-th group of $c_1n$ edges in that order to the $i$-th worker node, for all $i$. 
This is sufficient to simulate the randomized algorithm stated in Theorem~\ref{thm:mpc-mm} on the worker nodes, which takes $O(\log \log \Delta)$ rounds.
\end{proof}

For handling updates, we directly use our techniques from Section~\ref{sec:dyn-mm-ran}.

\begin{reptheorem}{thm:cc}
     \thmcc  
\end{reptheorem}

\begin{proof}
    By Lemma~\ref{lem:cc-mm}, we can find a maximal matching in the Congested Clique model within $O(\log \log \Delta)$ rounds. The update time complexities for an oblivious adversary and an adaptive adversary can be almost directly obtained from Theorem~\ref{thm:mm_random}. 
   
 {  The bottleneck for the update message complexity is the requirement that each node needs to know the states of its neighbors. Since in the initialization phase, each node knew the states of its neighbors, after~$\ell$ updates, we only need to use $O(\ell n)$ messages to update these information and this operation dominates the time complexity of other operations during our procedures. So, the update time complexity is proved.}
\end{proof}
\section{Basic Definitions from Information Theory}
\label{sec:tools}

In this section, we briefly recall some basic definitions and results from the information theory that we use in the paper.
Assume that $X$, $Y$, and $Z$ are discrete random variables.
Throughout, we use uppercase letters to denote random variables and corresponding lowercase letters to represent their values.

\begin{definition} \label{def:entropy}
  The \emph{entropy of $X$} is defined as
  \begin{align}
    \HH[ X ] = \sum_x \Pr[ X \!=\! x] \log_2(1 /\Pr[ X \!=\! x]). \label{eq:entropy}
  \end{align}
  The \emph{conditional entropy of $X$ conditioned on $Y$} is given by
  \begin{align}\label{eq:conditional_entropy}
    \HH[ X \mid Y ] &= \EE_y[ \HH[ X \mid Y \!=\! y] ] \\
      &= \sum_{y}^{} \Pr[ Y \!=\! y] \HH[ X \mid Y \!=\! y].\notag
  \end{align}
\end{definition}
Note that 
\begin{align}
	\HH\lt[ X \rt] \le \log_2\lt( \text{supp}(X) \rt),
\label{eq:uniform}
\end{align}
  where $\text{supp}(X)$ denotes the support of $X$, and equality is attained in \eqref{eq:uniform} if and only if $X$ is uniformly distributed on $\text{supp}(X)$.

\begin{definition} \label{def:mutual}
  Let $X$, $Y$, and $Z$ be discrete random variables.
  The \emph{conditional mutual information of $X$ and $Y$} is defined as
  \begin{align}
    \II[ X : Y \mid Z ]
      &= \HH[ X \mid Z ] - \HH[ X \mid Y, Z ] \label{eq:mutual_prop2}.
  \end{align}
\end{definition}

\begin{lemma} \label{lem:mutual_entropy}
  $\II[ X : Y \mid Z ] \le \HH[ X \mid Z ] \le \HH[ X ]$.
\end{lemma}

\begin{lemma}[Data Processing Inequality, see Theorem 2.8.1 in \cite{cover1999elements}] \label{lem:dataprocessing}
  If random variables $X$, $Y$, and $Z$ form the Markov chain $X \ra Y \ra Z$, i.e., the conditional distribution of $Z$ depends only on $Y$ and is conditionally independent of $X$, then
  \[
    \II[ X : Y ] \ge \II[ X : Z ].
  \]
\end{lemma}

\newpage

\end{document}